\documentclass[envcountsame, a4paper]{llncs}

\usepackage{amssymb}
\usepackage{amsmath}
\usepackage{graphicx}
\usepackage[hyphens]{url}
\usepackage{microtype}
\usepackage{color}
\usepackage{tikz}
\usetikzlibrary{arrows,automata}
\usepackage[wider,longer]{rmpage}

\newcommand{\first}{\texttt{first}}
\newcommand{\Nat}{\mathbb N}

\newcommand{\nan}{\vdash}
\newcommand{\bott}{\bot}
\newcommand{\undef}{\vdash}

\newcommand{\RealPlus}{{\mathbb R}_{> 0}}
\newcommand{\Rplus}{{\mathbb R}_{\geq 0}}
\newcommand{\RealZPlus}{{\mathbb R}_{\geq 0}}
\newcommand{\RealPlusBot}{{\mathbb R}_{>0}^{\nan}}

\newcommand{\AlphabetSet}{\Sigma}
\newcommand{\TAlphabetSet}{T{\kern-2pt}\Sigma}

\newcommand{\ERPred}[1]{\lhd_{#1}{\kern-1pt}}    % Event recording predicate symbol <|
\newcommand{\EPPred}[1]{\rhd_{{\kern-2pt}#1}{\kern-1pt}}    % Event predicting predicate symbol |>

\newcommand{\Suggestion}[1]{#1}%\textcolor{red}{#1}}
\renewcommand{\epsilon}{\varepsilon}

\newcommand{\cP}{\mathcal{P}}

\newcommand{\calI}{\mathcal{I}}

\newcommand{\barbbv}{\bar{\mathbb{V}}}
\newcommand{\barw}{\bar{w}}
\newcommand{\barm}{\bar{m}}
\newcommand{\baru}{\bar{u}}

\newcommand{\barv}{\bar{v}}

\newcommand{\bara}{\bar{a}}
\newcommand{\bart}{\bar{t}}
\newcommand{\barkappa}{\bar{\kappa}}

\newcommand{\ECMVPL}{\textsf{ECMVPL}}
\newcommand{\dtMVPA}{\textsf{dtMVPA}}

\newcommand{\ECA}{\textsf{ECA}}
\newcommand{\VPA}{\textsf{VPA}}
\newcommand{\MVPA}{\textsf{MVPA}}

\newcommand{\ECVPA}{\textsf{ECVPA}}
\newcommand{\ECMVPA}{\textsf{ECMVPA}}

\newcommand{\eca}{\textsf{ECA}}

\newcommand{\dtvpa}{\textsf{dtVPA}}
\newcommand{\ecvpa}{\textsf{ECVPA}}

\newcommand{\dtecvpa}{\textsf{dtECVPA}}

\newcommand{\dtmvpa}{\textsf{dt-MVPA}}
\newcommand{\ecmvpa}{\textsf{ECMVPA}}
\newcommand{\ecmvpl}{\textsf{ECMVPL}}
 \newcommand{\dtecmvpa}{\textsf{dt-ECMVPA}}
\newcommand{\dtecmvpl}{\textsf{dt-ECMVPL}}
\newcommand{\name}{\textsf}

\newcommand{\concept}{\textbf}

\newcommand{\andover}{\displaystyle \bigwedge}
\newcommand{\orover}{\displaystyle \bigvee}

\newcommand{\prodover}{\displaystyle \Pi}

\newcommand{\Iff}{\mbox{~iff~}}

\newcommand{\step}[1]{\xrightarrow{#1}}

\newcommand{\ldot}{{\rm <}\kern-0.37em{\raisebox{.25ex}{\bf .}}\kern0.375em}
 
 \usepackage{url,hyperref}
\newcommand{\set}[1]{\left\{ #1 \right\}}

\newcommand{\seq}[1]{\langle #1 \rangle}

\newcommand{\ram}[1]{{#1}}

\begin{document}

\title{On Timed Scope-bounded Context-sensitive Languages}
\author{
   D. Bhave\inst{1} 
   \and S. N. Krishna\inst{1}
   \and R. Phawade\inst{2}
   \and A. Trivedi\inst{3}
\authorrunning{Bhave, Krishna, Phawade, and Trivedi}
\institute{${}^1$ IIT Bombay and  
${}^2$ IIT Dharwad
${}^3$ CU Boulder
\email{$\{$devendra,krishnas$\}$@cse.iitb.ac.in,prb@iitdh.ac.in, ashutosh.trivedi@colorado.edu}
}
}
\maketitle

\begin{abstract}
    In (DLT 2016) we studied timed context sensitive
    languages characterized by multiple stack push down automata (MPA), 
    with an explicit bound on number of stages where in each stage  
    at most one stack is used ($k$-round MPA).

    In this paper, we continue our work on timed MPA  
    and study a subclass in which a symbol corresponding to
    a stack being pushed in it must be popped within fixed number of contexts of that stack
    ---scope-bounded push-down automata with multiple stacks ($k$-scope MPA).  
    We use Visibly Push-down Alphabet and Event Clocks to show that
    timed $k$-scope MPA have decidable reachability problem; are closed under  
    Boolean operations; and have an equivalent logical characterization.
\end{abstract}
 
\section{Introduction}

The Vardi-Wolper~\cite{VARDI19941} recipe for an automata-theoretic model-checking for a class of
languages requires that class to be closed under Boolean operations and 
have decidable emptiness problem.
Esparza, Ganty, and Majumdar \cite{EGM12} coined the term ``perfect languages'' for the
classes of languages satisfying these properties.
However, several important extensions of regular languages, such as pushdown
automata and timed automata,  do not satisfy
these requirements.
In order to lift the automata-theoretic model-checking framework for these classes
of languages, appropriate restrictions have been studied including visibly
pushdown automata~\cite{AM04vpda} (VPA) and event-clock automata~\cite{AFH99}
(ECA).
Tang and Ogawa~\cite{VTO09} introduced a perfect class of timed context-free
languages generalized both visibly pushdown automata and event-clock
automata to introduce event-clock visibly pushdown automata (ECVPA).

In this paper we study a previously unexplored class of timed context-sensitive 
languages inspired by the scope-bounded restriction on multi-stack visibly
pushdown languages introduced by La Torre, Napoli, and Parlato~\cite{TNP14}, and
show that it is closed under Boolean operations and has decidable emptiness
problem.
Moreover, we also present a logical characterization for the proposed subclass.

\noindent 
\textbf{Visible Stack Operations.}
Alur and Madhusudan~\cite{AM04vpda} introduced visibly pushdown automata as a
specification formalism where the call and return edges are made visible in a
structure of the word.
This notion is formalized by giving an explicit partition of the alphabet into
three disjoint sets of call, return, and internal or local symbols and the
visibly pushdown automata must push one symbol to the stack while reading a call
symbol, and must pop one symbol (given the stack is non-empty) while reading a
return symbol, and must not touch the stack while reading an 
internal symbol.

\noindent\textbf{Visible Clock Resets.}
Alur-Dill timed automata~\cite{AD90} is a generalization of finite automata with
continuous variables called clocks that grow with uniform rate in each control
location and their valuation can be used to guard the transitions.
Each transition can also reset clocks, and that allows one to constrain
transitions based on the duration since a previous transition has been taken.
However, the power of reseting clocks contributed towards timed automata not
being closed under complementation. 
In order to overcome this limitation, Alur, Fix, and Henzinger~\cite{AFH99}
introduced event-clock automata where input symbol dictate the resets of the
clocks.
In an event-clock automata every \ram{symbol} $a$ is implicitly associated with two
clocks $x_a$ and $y_a$, where the “recorder” clock $x_a$ records the time since the
last occurrence of the \ram{symbol} $a$, and the “predictor” clock $y_a$ predicts
the time of the next occurrence of \ram{symbol} $a$.
Hence, event-clock automata do not permit explicit reset of clocks and it is
implicitly governed by the input timed word.

\noindent\textbf{Visible Stack Operations and Clock Resets in Multistack Setting.}
We study dense-time event-clock multistack visibly pushdown automata  (\dtecmvpa{}) that combines
event-clock dynamics of event-clock automata with multiple visibly pushdown
stacks. We assume a partition of the alphabet among various stacks, and partition of the
alphabet of each stack into call, return, and internal symbols.
Moreover, we associate recorder and predictor clocks with each symbol.
Inspired by Atig et al.~\cite{AAS12} we consider our stacks to be dense-timed,
i.e. we allow stack symbols to remember the time elapsed since they were pushed
to the stack.

%A finite timed word over an alphabet $\Sigma$ is a sequence $(a_1,t_1), (a_2,t_2), \ldots,
A finite timed word over an alphabet $\Sigma$ is a sequence $(a_1,t_1), \ldots,
(a_n,t_n) \in (\Sigma {\times} \RealZPlus)^*$ such that $t_i \leq t_{i+1}$ for
all $1 \leq i \leq n-1$.    
Alternatively,  we can represent timed words as tuple $(\seq{a_1,\ldots,
a_n}, \seq{t_1, \ldots, t_n})$. 
We may use both of these formats depending on the context and for technical convenience. 
Let $T\Sigma^*$ denote the set of finite timed words over $\Sigma$. 
 
We briefly discuss the concepts of rounds and scope as introduced by~\cite{TNP14}.
Consider an pushdown automata with $n$ stacks.  
We say that for a stack $h$, a (timed) word is a stack-$h$ context if all of
its \ram{symbols} belong to the alphabet of stack $h$. 
A \emph{round} is fixed sequence of exactly $n$ contexts one for each stack. 
Given a timed word, it can be partitioned into sequences of contexts of various stacks.
The word is called $k$-round if it can be partitioned into \ram{$k$ rounds.}
We say that a timed word is $k$-scoped if for each return symbol of a
stack its matching call symbol occurs within the last $k$ contexts of that stack.
A visibly-pushdown multistack event-clock automata is
\emph{scope-bounded} if all of the accepting words are $k$-scoped for a fixed $k \in \Nat$.

\begin{figure}[t]
  \begin{center}\scalebox{0.8}{
      \begin{tikzpicture}[->,>=stealth',shorten >=1pt,%
auto,node distance=2.8cm,semithick,inner sep=3pt,bend angle=45,scale=0.97]
  \tikzstyle{every state}=[circle,fill=black!0,minimum size=3pt]
  \node[initial,state, initial where=above]     (0) at(0, 0)  {$l_0$};
  \node[state]             (1) at(2.5, 0)  {$l_1$};
  \node[state]             (2) at(5, 0) {$l_2$};
%  \node[state,fill=black!15]             (3) at(8, 0) {$l_3$};
%  \node[state,fill=black!15]   (4) at(12.2,0) {$l_4$};
  \node[state]             (3) at(8, 0) {$l_3$};
  \node[state]             (4) at(12.2,0) {$l_4$};
  \node[state]				 (5) at(12.2,-2) {$l_5$};
  \node[state]				 (6) at(12.2,-4) {$l_6$};
  \node[state]				 (7) at(8,-4) {$l_7$};
  \node[state]				 (8) at(5,-4) {$l_8$};
  \node[state]				 (9) at(2.5,-4) {$l_9$};
  \node[state,accepting]		 (10) at(0,-4) {$l_{10}$};

  \path 
    (0) edge node [above] {$\hat{a}$, push$^1$($\$$)} (1)
 
  (1) edge [loop above] node [above] {$~~a$, push$^1$($\alpha$)} (1)
  (1) edge node [above] {$b$, push$^2$($\$$)} (2)

  (2) edge [loop above] node [above] {$~~b$, push$^2$($\$$)} (2)
  (2) edge node [above] {$d$, pop$^2$($\$$)} (3);
 
    \path  (3) edge [loop above] node [above] {$d$, pop$^2$($\$$) } (3)
           (3) edge node [above] {$d$, pop$^2$($\$$) ${\in} [4,\infty)$} (4);
 
    \path (7) edge node [above] {$c$, pop$^1$($\$$)} (8);
     
    \path (8) edge node [above] {$\hat{c}$, pop$^1$($\alpha$)} (9);
    \path (8) edge [loop above] node [above] {$c$, pop$^1$($\$$)} (8);
    \path (8) edge node [below] {$x_{\hat{a}} \leq 5$} (9);

    \path (9) edge [loop above] node [above] {$d$, pop$^2$($\$$)} (9);
    \path (9) edge node [above] {$d$, pop$^2$($\$$)} (10);
    \path (9) edge node [below] {$x_b \leq 2 $} (10);
 
    \path (7) edge [loop above] node [above] {$b$, push$^2$($\$$)} (7);
     
    \path (6) edge node [above] {$b$, push$^2$($\$$)} (7);
    \path (6) edge node [below] {$y_d \leq 2$} (7);

    \path (4) edge [loop above] node [above] {$c$, pop$^1$($\$$)} (4);
    \path (4) edge node [left] {$c$, pop$^1$($\$$)} (5);
    \path (5) edge node [right] {$a$, push$^1$($\$$)} (6);
%    \path (6) edge [loop right] node [right] {$a$, pop$^1$($\$$)} (6);
    \path (6) edge [loop right] node [right] {$a$, \ram{push}$^1$($\$$)} (6);
\end{tikzpicture} 
}
\end{center}
\caption{Dense-time Multistack Visibly Pushdown Automata used in Example \ref{lab:ex:lan1}}
\label{fig:dtvpa}
\end{figure}
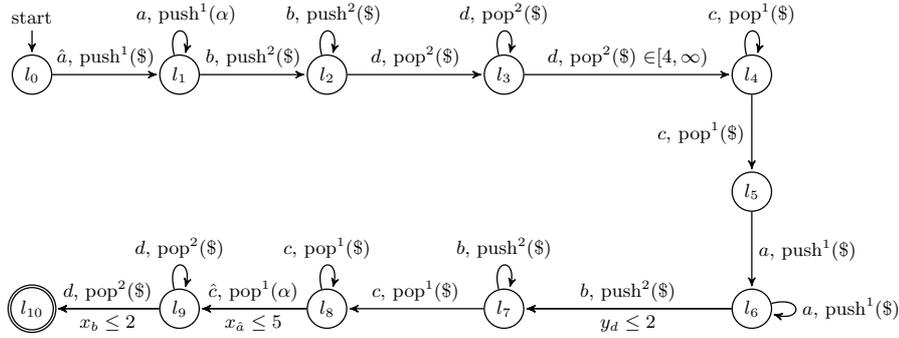

To introduce some of the key ideas of our model, let us consider the following example.

\begin{example}  
  \label{lab:ex:lan1}
  Consider the timed language whose untimed component is of the form \\
  $\ram{L=\{\hat{a}a^xb^yd^y c^la^l b^z c^x \hat{c}d^z \mid x,l,z \geq
1, ~y \geq 2\}}$
%  \[
%  \ram{L=\{\hat{a}a^xb^yd^y c^la^l b^z c^x \hat{c}d^z \mid x,l,z \geq 1, ~y \geq 2\}}
%  \]
  with the critical timing restrictions among various symbols in the following manner. 
  The time delay between the first occurrence of $b$ and the last occurrence of $d$
  in the substring $b^yd^y$ is at least $4$ time-units. 
  The time-delay between this last occurrence of $d$ and the next occcurence of
  $b$ is at most $2$ time-units.
  %  When second stack is emptied for the first time 
 %  within $4$ time units,  it should start re-filling it within 
 %  $2$ time units of it getting emptied.  
   Finally the last $d$ \ram{of the input string} must appear within $2$ time units of the last $b$, and
   $\hat{c}$ must occur within $5$ time units of corresponding $\hat{a}$. 
   This language is accepted by a \dtecmvpa{} with two stacks shown in
   Figure~\ref{fig:dtvpa}.
   We annotate a transition with the \ram{symbol} and corresponding stack operations
   if any. We write $pop^i$ or $push^i$ to emphasize pushes and pops to the
   $i$-th stack. We also use $pop^{i}(X) \in I$ to check if the age of the popped symbol $X$  
   belongs to the interval $I$. In addition, we use simple constraints on  
   predictor/recorder clock variables corresponding to the \ram{symbols}.
  Let $a,\hat{a}$ and $c,\hat{c}$ ($b$ and $d$, resp.) be call and return symbols for the
  first (second, resp.) stack.  
  The Stack alphabet for the first stack is  $\Gamma^1=\{\alpha,\$\}$ and for the second
  stack is $\Gamma^2=\{\$\}$.
  In Figure~\ref{fig:dtvpa} clock $x_a$ measures the time since the occurrence of the
  last $a$, while constraints $pop(\gamma) \in I$ checks if the age of the popped
  symbol $\gamma$ is in a given interval $I$.
%\cmnt{fix final states, and language: Reviewer 2 : minor issues} 
This language is $3$-scoped and is accepted by a $6$-round
    \dtecmvpa{}.
But if we consider the Kleene star of this language, it will be still $3$-scoped.   
and its machine can be built by fusing states $l_0$ and $l_{10}$ of the MVPA  in Figure~\ref{fig:dtvpa}.
\end{example}

\noindent\textbf{Related Work.}
The formalisms of timed automata and pushdown stack have been combined before.
First such attempt was timed pushdown automata~\cite{Bouajjani1995} by
Bouajjani, et al and was proposed as a timed extension of pushdown automata
which uses global clocks and timeless stack. 
We follow the  dense-timed pushdown automata by Abdulla
et al~\cite{AAS12}.
The reachability checking of a given location from an intial one was shown to be
decidable for this model. 
Trivedi and Wojtczak~\cite{TW10} studied the recursive timed automata in which
clock values can be pushed onto a stack using mechanisms like pass-by-value and
pass-by-reference. They studied reachability and termination problems for this
model. 
Nested timed automata (\textsf{NeTA}) proposed by Li, et al~\cite{Li2013} is a
relatively recent model which, an instance of timed automata itself can be
pushed on the stack along with the clocks. 
The clocks of pushed timed automata progress uniformly while on the stack. 
From the perspective of logical characterization, timed matching logic, an
existential fragment of second-order logic, identified by Droste and
Perevoshchikov~\cite{DP15} characterizes dense-timed pushdown automata. 
We earlier~\cite{BDKRT15} studied MSO logic for dense-timed visibly pushdown
automata which form a subclass of timed context-free languages. 
This subclass is closed under union, intersection, complementation and
determinization. 
The work presented in this paper extends the results from~\cite{BDKPT16} for
bounded-round \dtecmvpa{} to the case of bounded-scope \dtecmvpa{}.

\noindent 
\textbf{Contributions.} We study bounded-scope \dtecmvpa{}
and show that they are closed \ram{under}
Boolean operations and the emptiness problem for these models is decidable.
We also present a logical characterization for these models.

\noindent
\concept{Organization of the paper:}
%In the next section we recall the definitions of event clock automata and visibly pushdown automata.
In the next section we recall the definitions of event clock and visibly pushdown automata.
In Section \ref{lab:sec:def} we define $k$-scope dense time multiple
stack visibly push down automata with event clocks and its properties.
In the following section these properties are used to decide emptiness  
checking and determinizability of $k$-scope \ecmvpa{} with event clocks.
Building upon these results, %the Section~\ref{lab:sec:det-kecmvpa}  
%we show decidability of emptiness checking and determinizability of 
we show decidability of these properties for %emptiness checking and determinizability of 
$k$-scope \dtecmvpa{} with event clocks.
%where stack elements can have time stamps from reals is obtained. 
In Section \ref{sec:mso} we give a logical characterization for models
introduced.
%before concluding in the last section. 
%before giving future directions and concluding in the last section. 

\section{Preliminaries}
\label{lab:sec:prelim}

We only give a very brief introduction of required
concepts in this section, and for a detailed background on these concepts we
refer the reader to~\cite{AD94,AFH99,AM04vpda}. 
We assume that the reader is comfortable with standard concepts  
such as context-free languages, pushdown automata, MSO logic from 
automata theory; and clocks, event clocks, clock constraints, and valuations 
from timed automata. Before we introduce our model, we revisit the definitions  
of event-clock automata. %% and visibly pushdown automata. 
%of event-clock automata and visibly pushdown automata. 

\subsection{Event-Clock Automata}
The general class of TA \cite{AD94} are not closed under Boolean
operations. An important class of TA which is determinizable is 
Event-clock automata (\ECA{})~\cite{AFH99}, and hence closed under
Boolean operations. Here the
determinizability is achieved by making clock resets ``visible''.
%so that in an input timed word, the clock valuations after reading a  
%prefix of the word are also determined by the timed word. 
 
To make clock resets visible we have two clocks which are
%To make clock resets visible in ECA we have two clocks which are
%implicitly associated with every action $a \in \Sigma$ : 
associated with every action $a \in \Sigma$ : 
$x_a$ the ``recorder'' clock which records the time of the last  
occurrence of action $a$, and $y_a$ the ``predictor'' clock 
which predicts the time of the next occurrence of action $a$. 
For example, for a timed word $w = (a_1,t_1), (a_2,t_2), \dots, (a_n,t_n)$, the value 
of the event clock $x_{a}$ at position $j$ is $t_j-t_i$ where $i$ is the
largest position preceding $j$ where an action $a$ occurred. 
If no $a$ has occurred before the $j$th position, then the
value of $x_{a}$ is undefined denoted by a special symbol $\nan$.
Similarly, the value of $y_{a}$ at position $j$ of $w$ is undefined if symbol
$a$ does not occur in $w$ after the $j$th position. 
Otherwise, it is $t_k-t_j$ where $k$ is the first occurrence of $a$
after $j$. 
 
Hence, event-clock automata do not permit explicit reset of clocks and it is
implicitly governed by the input timed word which makes them 
determinizable and closed under all Boolean operations.  
 
We write $C$ for the set of all event clocks and we use $\RealPlusBot$ for the
set $\Suggestion{\RealPlus} \cup \{\nan\}$.
Formally, the clock valuation after reading $j$-th prefix of the input timed
word $w$, $\nu_j^w: C \mapsto \RealPlusBot$, is defined as follows:
$\nu_j^w(x_q) =  t_j {-} t_i$ if there exists an $0 {\leq} i {<} j$ such that 
$a_i = q $ and  $a_k \not = q$ for all $i {<} k {<} j$, otherwise
$\nu_j^w(x_q) = \ \nan$ (undefined). 
Similarly,  $\nu_j^w(y_q) = t_m - t_j$ if there is $j {<} m$ such that  $a_m = q$ and
    $a_l \not = q$  for all $j {<} l {<} m$, otherwise $\nu_j^w(y_q) =  \nan$.

A clock constraint over $C$ is a boolean combination of constraints of the
form $z \sim c$ where $z \in C$, $c \in \Nat$ and $\sim \in \{\leq,\geq\}$.  
Given a clock constraint $z \sim c$ over $C$, we write $\nu_i^w \models (z \sim c)$ to denote 
if $\nu_j^w(z) \sim c$. For any boolean combination $\varphi$,  
$\nu_i^w \models \varphi$ is defined in an obvious way: if $\varphi=\varphi_1 \wedge \varphi_2$, then 
$\nu_i^w \models \varphi$ iff $\nu_i^w \models \varphi_1$ and $\nu_i^w \models \varphi_2$. Likewise, the 
other Boolean combinations are defined.  
 
\ram{Let $\Phi(C)$ define all the
clock constraints defined over $C$.}

\begin{definition}
  An event clock automaton is a tuple $A = (L, \Sigma, L^0, F, E)$  where $L$ is a
set of finite locations, $\Sigma$ is a  finite alphabet, $L^0 \in L$ is the set of
initial locations, $F \in L$ is the set of final locations, and $E$ is  
a finite set of edges of the form $(\ell, \ell', a, \varphi)$ where $\ell, \ell'$ are
locations, $a \in \Sigma$, and $\varphi$ \ram{in $\Phi(C)$}. 
\end{definition}
The class of languages accepted by \eca{} have Boolean closure
and decidable emptiness~\cite{AFH99}.

\subsection{Visibly Pushdown Automata}
%\textbf{Visibly Pushdown Alphabet}: 
The class of push down automata are not determinzable and also not
closed under Boolean operations \cite{HU79}.  
The determinizability is achieved by making input alphabet
``visible'' that is for a given input letter only one kind of stack
operations is allowed giving an important subclass of    
Visibly pushdown automata~\cite{AM04vpda} which 
operate over words that dictate the stack operations.
This notion is formalized by giving an explicit partition of the
alphabet. 
This notion is formalized by giving an explicit partition of the alphabet into
three disjoint sets of \emph{call}, \emph{return}, and \emph{internal} symbols and
the visibly pushdown automata must push one symbol to stack while reading a call
symbol, and must pop one symbol (given stack is non-empty)  while reading a
return symbol, and must not touch the stack while reading the internal symbol.

\begin{definition}
%A visibly pushdown alphabet is a tuple $\Sigma = \seq{\Sigma_c, \Sigma_r, \Sigma_{l}}$
A visibly pushdown alphabet is a tuple $\seq{\Sigma_c, \Sigma_r, \Sigma_{l}}$
where $\Sigma_c$ is  \emph{call} alphabet,  
$\Sigma_r$ is a \emph{return} alphabet, and 
$\Sigma_{l}$ is \emph{internal} alphabet. 
 
A \name{visibly pushdown automata(VPA)} over $\Sigma = \seq{\Sigma_c,
\Sigma_r, \Sigma_{l}}$ is a tuple
$(L, \Sigma, \Gamma, L^0, \delta, F)$ where $L$ is a finite set of locations
including a set $L^0 \subseteq L$ of initial locations, $\Gamma$ is a finite
stack alphabet with special end-of-stack symbol $\bott$, $\Delta \subseteq 
(L {\times} \Sigma_c {\times} L {\times} (\Gamma {\setminus} {\bott})) \cup (L {\times}
\Sigma_r {\times}\Gamma {\times} L) \cup (L {\times} \Sigma_{l}{\times} L)$ is
the transition relation, and $F \subseteq L$ is final locations.
\end{definition}
 
Alur and  Madhusudan~\cite{AM04vpda} showed that \VPA{}s are
determinizable and closed under boolean operations.
A language $L$ of finite words defined over visibly pushdown alphabet $\Sigma$ is  
a \name{visibly pushdown language(VPL)} if there exist a VPA $M$ such that $L(M)=L$. 
The class of languages accepted by visibly pushdown automata are closed under boolean
operations with decidable emptiness property~\cite{AM04vpda}.

\section{Dense-Time Visibly Pushdown Multistack Automata}
\label{lab:sec:def}

This section introduces scope-bounded dense-timed multistack 
visibly pushdown automata and give some properties 
about words and languages accepted by these machines. 
 
Let $\AlphabetSet= \seq{\Sigma^h_c, \Sigma^h_r, \Sigma^h_{l}}_{h=1}^n$  
where $\Sigma^i_x \cap \Sigma^j_x=\emptyset$
whenever either $i \neq j$ or $x \neq y$, and $x,y \in \{c,r,l\}$.
Let $\Sigma^h=\seq{\Sigma^h_c, \Sigma^h_r, \Sigma^h_{l}}$. 
Let $\Gamma^h$ be the stack alphabet of the $h$-th stack and 
$\Gamma=\bigcup_{h=1}^n \Gamma^h$.  
For notational convenience, we assume that  
%the partitioning function is one-to-one, i.e.  
each symbol $a \in \Sigma^h$ has an unique recorder $x_a$ and
predictor $y_a$ clock assigned to it. 
Let $C_h$ denote the set of event clocks corresponding to stack $h$ and 
$\Phi(C_h)$ denote the set of clock constraints defined over $C_h$.
Let $cmax$ be the maximum constant used in the clock constraints $\Phi(C^h)$  
of all stacks. % corresponding to stack $h$ of given \ram{\dtecmvpa{}}.  
Let $\calI$  denote the finite set of intervals  
$\{ [0,0],(0,1),[1,1], (1,2), \ldots,[cmax,cmax],(cmax,\infty)\}$.
\begin{definition}[\cite{BDKPT16}] A dense-timed visibly pushdown
multistack automata (\dtecmvpa{}) over
%$\seq{\Sigma^h_c, \Sigma^h_r, \Sigma^h_{int}}_{h=1}^{n}$ is a tuple  
$\seq{\Sigma^h_c, \Sigma^h_r, \Sigma^h_{l}}_{h=1}^{n}$ is a tuple  
%$(L, \Sigma, \Gamma, L^0, F, \Delta {=} (\Delta^h_c {\cup} \Delta^h_r {\cup} \Delta^h_{int})_{h=1}^n)$ 
$(L, \Sigma, \Gamma, L^0, F, \Delta {=} (\Delta^h_c {\cup} \Delta^h_r {\cup} \Delta^h_{l})_{h=1}^n)$ 
where
\begin{itemize}
\item $L$ is a finite set of locations
including a set $L^0 \subseteq L$ of initial locations,
\item
$\Gamma^h$ is the finite alphabet of stack $h$ and has special end-of-stack symbol $\bott_h$, 
\item
%\cmnt{define $\Phi(C_h)$ and $C_h$}
$\Delta^h_c \subseteq (L {\times} \Sigma^h_c {\times} \Phi(C_h) {\times} L {\times} (\Gamma^h
{\setminus} \{\bott_h\}))$ is the set of call transitions,  
\item
$\Delta^h_r \subseteq (L {\times} \Sigma^h_r {\times} \calI {\times} \Gamma^h {\times} \Phi(C_h)
{\times} L)$ is set of return transitions, 
\item
%$\Delta^h_{int} \subseteq (L {\times} \Sigma^h_{int} {\times} \Phi(C_h) {\times} L)$ is
$\Delta^h_{l} \subseteq (L {\times} \Sigma^h_{l} {\times} \Phi(C_h) {\times} L)$ is
set of internal transitions, 
and
\item $F {\subseteq} L$ is the set of final locations.
\end{itemize}
\end{definition}         
%\paragraph{Semantics}
Let $w = (a_0,t_0), \dots, (a_e,t_e)$ be a timed word.  
A configuration of the \ram{\dtecmvpa{}} is a tuple
$(\ell, \nu_i^w, (((\gamma^1\sigma^1, age(\gamma^1\sigma^1)), \dots, (\gamma^n\sigma^n, age(\gamma^n\sigma^n)))$ where   
$\ell$ is the current location of the \ram{\dtecmvpa{}}, function $\nu_i^w$ gives the valuation of
all the event clocks at position $i \leq |w|$, $\gamma^h\sigma^h \in \Gamma^h(\Gamma^h)^*$
is the content of stack $h$ with $\gamma^h$ being the topmost symbol,
and $\sigma^h$ the string representing stack contents below $\gamma^h$,       
while $age(\gamma^h\sigma^h)$ is a sequence of real numbers denoting
%while \ram{$age(\gamma^h\sigma^h)$ is a sequence of real numbers denoting
the ages (the time elapsed since a stack symbol was pushed on to the stack)
of all the stack symbols in $\gamma^h \sigma^h$.
We follow the assumption that $age(\bott^h) = \seq{\undef}$ (undefined).
If for some string $\sigma^h \in (\Gamma^h)^*$ we have $age(\sigma^h) = \seq{t_1,
t_2, \ldots, t_g}$ and for $\tau \in \Rplus$ then we write $age(\sigma^h) + \tau$ for the sequence
$\seq{t_1+ \tau, t_2+ \tau, \ldots, t_g+\tau}$. 
For a sequence $\sigma^h = \seq{\gamma^h_{1}, \ldots, \gamma^h_{g}}$
and a stack symbol
$\gamma^h$ we write $\gamma^h::\sigma^h$ for $\seq{\gamma^h, \gamma^h_{1}, \ldots, \gamma^h_{g}}$. 
 
A run of a \ram{\dtecmvpa{}} on a timed word $w = (a_0,t_0), \dots, (a_e,t_e)$
is a sequence of configurations:  

\noindent 
$(\ell_0, \nu^{\Suggestion{w}}_0, 
(\seq{\bott^1},\seq{\undef}), \dots, (\seq{\bott^n},\seq{\undef}))$,
$
(\ell_1, \nu^{\Suggestion{w}}_1,((\sigma^1_1, age(\sigma^1_1)), \dots,
(\sigma^n_1, age(\sigma^n_1))))$,\\
$\dots, (\ell_{e+1}, \nu^{\Suggestion{w}}_{e+1}, (\sigma^{1}_{e+1},
age(\sigma^{1}_{e+1})), \dots, (\sigma^{n}_{e+1},
age(\sigma^{n}_{e+1})))$
where $\ell_i \in L$,
$\ell_0 \in L^0$,
\ram{$\sigma^h_i \in (\Gamma^h)^* \bott^h$}, and
%$\sigma^h_i \in (\Gamma^h \cup \{\bott^h\})^+$, and
for each $i,~0 \leq i \leq e$, we have:
 \begin{itemize}  
 \item
   If $a_i \in \Sigma^h_c$, then there is $(\ell_i, a_i,\varphi,
   \ell_{i+1}, \gamma^h) {\in} \Delta^h_c$ such that
   $\nu_i^w \models \varphi$. 
   The symbol $\gamma^h \in \Gamma^h \backslash \{\bott^h\}$ is then pushed 
   onto the stack $h$, and its age is initialized to zero, i.e. 
   $(\sigma^{h}_{i+1}, age(\sigma^h_{i+1}))= (\gamma^h ::\sigma^h_i,
   0::(age(\sigma^h_i)+(t_i-t_{i-1})))$.  
  % Note that all symbols in the stack $h$ excluding the topmost age by $t_i-t_{i-1}$.
   All symbols in all other stacks are unchanged, and they age by $t_i-t_{i-1}$.
 \item
  If $a_i \in \Sigma^h_r$, then there is
  $(\ell_i, a_i, I, \gamma^h, \varphi, \ell_{i+1}) \in \Delta^h_r$
  such that $\nu_i^w \models \varphi$. 
  Also, $\sigma^h_{i} = \gamma^h :: \kappa \in \Gamma^h (\Gamma^h)^*$ and 
  $age(\gamma^h)+(t_i-t_{i-1}) \in I$.
  The symbol $\gamma^h$ is popped from stack $h$ obtaining  $ \sigma^h_{i+1}=\kappa$ and 
  ages of remaining stack symbols are updated i.e.,  
\ram{ $age(\sigma^h_{i+1})=age(\kappa)+  (t_i-t_{i-1}) $. }
  However, if $\gamma^h= \seq{\bott^h}$, then $\gamma^h$  is not popped.  
%, and the %   attached interval $I$ is irrelevant.
The contents of all other stacks remains unchanged, and simply age by  $(t_i-t_{i-1})$. 

  \item
    %If $a_i {\in} \Sigma^h_{int}$, then there is $(\ell_i, a_i, \varphi,
If $a_i {\in} \Sigma^h_{l}$, then there is $(\ell_i, a_i, \varphi,
\ell_{i+1}) {\in} \Delta^h_{l}$ such that $\nu_i^w \vDash \varphi$. 
In this case all stacks remain unchanged i.e. $\sigma^h_{i+1}{=}\sigma^h_i$,  
\ram{but their contents age by $t_i-t_{i-1}$ i.e.  
$age(\sigma^h_{i+1}){=}age(\sigma^h_i)+(t_{i}-t_{i-1})$ for all $1
\leq h \leq n$.}
\end{itemize}  

A run $\rho$ of a {\dtecmvpa{}} $M$ is accepting if it terminates in a final location. 
A timed word $w$ is an accepting word if there is an accepting run of $M$ on $w$.
The language $L(M)$ of a \ram{\dtecmvpa{}} $M$, is the set of all  timed words $w$
accepted by $M$ \ram{and is called \dtecmvpl{}.}

%\paragraph{Deterministic \dtMVPA{}}
A \ram{\dtecmvpa{}} $M=(L, \AlphabetSet, \Gamma, L^0,
F, \Delta)$ is said to be \emph{deterministic} if it has exactly one
start location, and for every configuration  and input action exactly one
transition is enabled. Formally, we have the following conditions:
 \ram{ 
for any two moves $(\ell, a, \phi_1, \ell',\gamma_1)$ and $(\ell, a, \phi_2, \ell'',
	 \gamma_2)$ of $\Delta^h_c$, condition $\phi_1 \wedge \phi_2$ is unsatisfiable; 
for any two moves $(\ell, a, I_1, \gamma, \phi_1, \ell')$ and $(\ell, a, I_2, \gamma,
	 \phi_2, \ell'')$  in $\Delta^h_r$, either  $\phi_1 \wedge \phi_2$ is unsatisfiable  
	 or $I_1 \cap I_2 = \emptyset$; and
for any two moves $(\ell, a, \phi_1, \ell')$ and $(\ell, a, \phi_2, \ell')$  
	in $\Delta^h_{l}$, condition $\phi_1 \wedge \phi_2$ is unsatisfiable.}

% 
% 
%  for every $(\ell, a, \phi_1, \ell',\gamma_1), (\ell, a, \phi_2, \ell'',
%  \gamma_2) \in \Delta^h_c$, $\phi_1 \wedge \phi_2$ is unsatisfiable; 
%  for every $(\ell, a, I_1, \gamma, \phi_1, \ell'), (\ell, a, I_2, \gamma,
%  \phi_2, \ell'') \in \Delta^h_r$,
%  either  $\phi_1 \wedge \phi_2$ is unsatisfiable or $I_1 \cap I_2 =
%  \emptyset$; and
%  for every $(\ell, a, \phi_1, \ell'), (\ell, a, \phi_2, \ell') \in
%\Delta^h_{l}$, 
% $\phi_1 \wedge \phi_2$ is unsatisfiable.
%% \begin{enumerate}
%% \item 
%%   for every $(\ell, a, \phi_1, \ell',\gamma_1), (\ell, a, \phi_2, \ell'',
%%   \gamma_2) \in \Delta_c$, $\phi_1 \wedge \phi_2$ is unsatisfiable; 
%% \item 
%%   for every $(\ell, a, I_1, \gamma, \phi_1, \ell'), (\ell, a, I_2, \gamma,
%%   \phi_2, \ell'') \in \Delta_r$,
%%   either  $\phi_1 \wedge \phi_2$ is unsatisfiable or $I_1 \cap I_2 =
%%   \emptyset$; and
%% \item 
%%   for every $(\ell, a, \phi_1, \ell'), (\ell, a, \phi_2, \ell') \in \Delta_l$, 
%%  $\phi_1 \wedge \phi_2$ is unsatisfiable.
%%   \end{enumerate}  
\ram{An Event clock multi stack visibly push down automata ($\ECMVPA$) is a  
$\dtecmvpa{}$ where the stacks are untimed i.e., 
a \dtecmvpa{} $(L, \AlphabetSet, \Gamma, L^0, F, \Delta)$, with $I = [0, +\infty]$  
for every $(\ell, a, I, \gamma, \phi, \ell') {\in} \Delta^h_r$, is an \ECMVPA{}.} 
 
A $\dtecvpa{}$ is a $\dtecmvpa{}$ restricted to single stack. 
 
%An $\ECMVPA$ is a $\dtMVPA{}$ where the stacks are untimed. 
%A \ECMVPA{} $(L, \AlphabetSet, \Gamma, 
%L^0, F, \Delta)$ is an \dtMVPA{} if $I = [0, +\infty]$ for every $(\ell, a,
%I, \gamma, \phi, \ell') {\in} \Delta^h_r$.
 
We now define a \emph{matching relation} $\sim_h$ on the positions of input timed word $w$  which identifies matching call and return positions for each stack $h$.
Note that this is possible because of the visibility of the input symbols.
 
\begin{definition}[Matching relation]
Consider a timed word $w$ over $\Sigma$. 
\ram{Let $\cP^h_c$ (resp. $\cP^h_r$) denote the set of positions in $w$ where a symbol from  
$\Sigma^h_c$ i.e. a call symbol (resp. $\Sigma^h_r$ i.e. a return symbol ) occurs.
Position i (resp. j) is called \emph{call position (resp. return
position).}}
For each stack $h$ the timed word $w$, defines a \emph{matching relation}  
$\sim_h \subseteq \cP^h_c \times \cP^h_r$ satisfying the following conditions: 
%$(1)$ For all positions $i,j$ with $i \sim_h j$ we have $i < j$.
%$(2)$ For any call position $i$ of $\cP^h_c$ and any return position $j$ of $\cP^h_r$  with $i <j$, there exists $l$ with $i \leq l \leq  j$ for which either $i \sim_h l$ or $l \sim_h j$.  
%$(3)$ For each call position $i \in \cP^h_c$ (resp. $i \in \cP^h_r$) there is at most
%one return position $j \in \cP^h_r$ (resp. $j \in \cP^h_c$) with $i \sim_h j$ (resp. $j \sim_h i$).
 
\begin{enumerate}  
\item  for all positions $i,j$ with $i \sim_h j$ we have $i < j$,
\item  for any call position $i$ of $\cP^h_c$ and any return position
$j$ of $\cP^h_r$ 
with $i <j$, there exists $l$ with $i \leq l \leq  j$ for which either
$i \sim_h l$ or $l \sim_h j$,
\item for each call position $i \in \cP^h_c$ (resp. $i \in \cP^h_r$) there is at most
one return position $j \in \cP^h_r$ (resp. $j \in \cP^h_c$) with $i \sim_h j$ (resp. $j \sim_h i$).
\end{enumerate}

\end{definition} 
 
For $i \sim_h j$, position $i$ (resp. $j$) is called  
\emph{matching call (resp. matching return)}.
 
This definition of matching relation extends that defined by La Torre, et al~\cite{TNP16} to timed words.
As matching relation is completely determined by stacks and timestamps of the input word does not play any role, we claim that above definition uniquely identifies matching relation for a given input word $w$ using uniqueness proof from~\cite{TNP16}.
 
Fix a $k$ from $\Nat$.  
A \emph{stack-$h$ context} is a word in $\Sigma^h (\Sigma^h)^*$. 
Given a word $w$ and a stack $h$, the word $w$ has $k$ maximal $h$-contexts if 
%$w \in (\Gamma^h)^* ( (\Gamma^{h'})^*  (\Gamma^h)^* )^{k-1}$, where $h \neq h'$.
$w \in (\Sigma^h)^* ( (\bigcup_{h \neq h'} \Sigma^{h'})^*  (\Sigma^h)^* )^{k-1}$.
A timed word over $\Sigma$ is \concept{$k$-scoped} if for each
matching call of stack $h$, its corresponding return occurs within at most
$k$ maximal stack-$h$ contexts.
 
%A timed word over $\Sigma$ is \concept{$k$-scoped} if for each matching call its
%return occurs in at most $k$ stack-$h$ contexts. 
 
Let $Scope(\Sigma,k)$ denote the set of all $k$-scope timed words over $\Sigma$.  
For any fixed $k$, a $k$-scope \ram{\dtecmvpa{}}
%over $\Sigma$ is a tuple $A=(k, L, \AlphabetSet, \Gamma, L^0, F, \Delta)$  
over $\Sigma$ is a tuple \ram{$A=(k,M)$} where $M=(L, \AlphabetSet, \Gamma, 
L^0, F, \Delta)$ is a \ram{\dtecmvpa{}} over $\Sigma$.  The language accepted by $A$ is 
$L(A)=L(M) \cap Scope(\Sigma,k)$ and is called $k$-scope 
\ram{dense-timed multistack visibly pushdown language ($k$-scoped-\dtecmvpl{}).  
We define $k$-scoped-\ECMVPL{} in a similar fashion.}
%The class of $k$-scope dense-timed multistack visibly push down languages is denoted $k$-\dtMVPL. 
%We define  $k$-\ECMVPL{} and  $k$-\dtecmvpl{}  in a similar fashion.   

%\subsection{Bounded Number of Rounds}
%\paragraph*{Key definitions}
We now recall some key definitions from La Torre~\cite{TNP14,TNP16}
which help us extend the notion of scoped words from untimed to timed words.
%Let \mvalpha{} be a visibly pushdown alphabet and $w$ be a timed word over $\Sigma$.

%
%\begin{definition}[$k$-scoped splitting~\cite{TNP14,TNP16}]\mbox{}
%\vspace*{\parsep}
%\vspace*{0.5\baselineskip}
% \begin{itemize}
%  \item A \emph{cut} of $w$ is $w_1 {:} w_2$ where $w=w_1w_2$. The cutting of $w$ is marked by ``:''.
%  \item A cut is $h$-\emph{consistent} with respect to matching relation $\sim_h$ if no call occuring in $w_1$ matches with a return in $w_2$ in $\sim_h$.
%  \item A \emph{splitting} of $w$ is a set of cuts $w_1 \ldots w_i:w_{i+1}$
%such that $w=w_1\ldots w_i w_{i+1}\ldots w_m$.  
%  \item A \emph{h-consistent splitting} of $w$ is the one in which each  
%specified cut is $h$-consistent. 
%%  \item A \emph{stack-$h$ context} over $\Sigma^h=\seq{\Sigma^h_c, \Sigma^h_r, \Sigma^h_{l}}$  is a timed word in $(\Sigma^h)^*$. The empty word $\epsilon$ is a context.   
%  \item A \emph{context-splitting} of word $w$ is a splitting $w_1:w_2:\ldots:w_m$  
%such that each $w_i$ is an $h$-context for some stack $h$ and $i \in \{1,\ldots,m\}$. 
%  \item A \emph{canonical context-splitting} of  word is a context-splitting
%of $w$ in which no two consecutive contexts belong to the same stack.  
% \end{itemize} 
%\end{definition}
% 
 \begin{definition}[$k$-scoped splitting~\cite{TNP14,TNP16}]
   A \name{cut} of $w$ is $w_1 {:} w_2$ where $w=w_1w_2$. The cutting of $w$ is marked by ``:''.
   A \name{cut is $h$-consistent} with matching relation $\sim_h$ if no call occuring in $w_1$ matches with a return in $w_2$ in $\sim_h$.

   A \name{splitting} of $w$ is a set of cuts $w_1 \ldots
w_i:w_{i+1}\ldots w_m$ such that $w=w_1\ldots w_i w_{i+1}\ldots w_m$
for each $i$ in $\{1,\ldots,m-1\}$. 
   An \name{h-consistent splitting} of $w$ is the one in which each  
specified cut is $h$-consistent. 
   A \name{context-splitting} of word $w$ is a splitting $w_1:w_2:\ldots:w_m$  
such that each $w_i$ is an $h$-context for some stack $h$ and $i \in \{1,\ldots,m\}$. 
   A \name{canonical context-splitting} of  word is a context-splitting
of $w$ in which no two consecutive contexts belong to the same stack.  
\end{definition}

% Given a timed word $w$ over $\Sigma$ and stack $h$, a \emph{cut} of $w$ is $w_1:w_2$  
% where $w=w_1 w_2$, and it is called \name{$h$-consistent cut} if there does not  
% exist a stack-$h$ call of $w_1$ which has a matching return in $w_2$.   
 
Given a context-splitting of timed word $w$, we obtain its
\name{$h$-projection} by removing all non stack-$h$ contexts. 
See that an $h$-projection is a context-splitting. 
%This $h$-projection is \name{$k$-bounded} if 
%there exists a $h$-consistent splitting of this $h$-projection, where 
%each component of the cut is a concatenation of at most $k$ $h$-contexts. 
An ordered tuple of $m$ $h$-contexts is \name{$k$-bounded} if there
there exists a $h$-consistent splitting of this tuple, where 
each component of the cut in the splitting is a concatenation
of at most $k$ consecutive $h$-contexts of given tuple. 
A \name{$k$-scoped splitting} of word $w$ is the canonical splitting
of $w$ equipped with additional cuts for each stack $h$ such that, if we take
$h$-projection of $w$ with these cuts it is $k$-bounded. 

%Also, we write $k$-\dtMVPA{} and 
%$k$-\ECMVPA{} to denote $k$-scope \dtMVPA{} and  $k$-scope \ECMVPA. 
 
% A \emph{round} over $\Sigma$ is a timed word $w$ over $\Sigma$ of the form $w_1w_2 \dots w_n$  
% such that each $w_h$ is a context over $\Sigma^h$. A $k$-round over $\Sigma$ is a timed word $w$ 
% that can be obtained as a concatenation of $k$ rounds over $\Sigma$.  
% That is, $w=u_1 u_2 \dots u_k$, where each $u_i$ is a round. 
% Let $Round(\Sigma,k)$ denote the set of all $k$-round timed words over $\Sigma$.  
% For any fixed $k$, a $k$-round \dtMVPA{}
% over $\Sigma$ is a tuple $A=(k, L, \AlphabetSet, \Gamma, 
% L^0, F, \Delta)$ where $M=(L, \AlphabetSet, \Gamma, 
% L^0, F, \Delta)$ is a \dtMVPA{} over $\Sigma$.  The language accepted by $A$ is 
% $L(A)=L(M) \cap Round(\Sigma,k)$ and is called $k$-round 
% dense-timed multistack visibly pushdown language. The class of $k$-round 
% dense-timed multistack visibly  pushdown languages is denoted $k$-\dtMVPL. 
% The set $\bigcup_{k \geq 1}k$-\dtMVPL{} is denoted  
% $bd$-\dtMVPL, and is the class of dense-timed multistack visibly  pushdown languages  with a bounded number of rounds. 
% We define  $k$-\dtecmvpl{} and  \ecmvpl{} in a similar fashion.
 
The main purpose for introducing all the above definitions is to come up with a scheme which will permit us to split any arbitrary length input timed word into $k$-scoped words.
Using \cite{TNP14,TNP16} for untimed words we get the following
Lemma. 
\begin{lemma}A timed word $w$ is $k$-scoped iff there is a $k$-scoped splitting of $w$. 
    \label{lab:lm:kscope-spilt} 
\end{lemma}

Next we describe the notion of switching vectors for timed
words~\cite{BDKRT15}, which are used in determinization of $k$-scope
$\dtecmvpa{}$.   

\subsection{Switching vectors}
\label{lab:subsec:split-switch}
Let $A$ be $k$-scoped \dtecmvpa{} over $\Sigma$ and let $w$ be a timed word
accepted by $A$.  
Our aim is to simulate $A$ on $w$ by $n$ different \dtecvpa{}s, $A^h$  
for each stack-$h$ inputs.
We insert a special symbol $\#$ at the end of each
maximal context, to obtain word $w'$ over $\Sigma \cup \{ \#,\#'\}$. 
We also have recorder clocks $x_{\#}$ and predictor clocks $y_{\#}$ for symbol $\#$. 
For $h$-th stack, let \dtecvpa{} $A^h$ be the restricted version of
$A$ over alphabet $\Sigma \cup \{ \#, \#'\}$ which simulates $A$ on  
input symbols from $\Sigma^h$.
Then, it is clear that at the symbol before $\#$, stack $h$ may be touched
by $\dtecmvpa{}$ $A$ and at the first symbol after $\#$, stack $h$ may be  
touched again. But it may be the case that at positions where $\#$
occurs stack $h$ may not be empty i.e., cut defined position of \# may
be not be $h$-consistent. 

To capture the behaviour of $A^h$ over timed word $w$ we have a  
notion of switching vector.  
Let $m$ be the number of maximal $h$-contexts in word $w$ 
and $w^h$ be the $h$-projection of $w$ i.e., $w^h=u^h_1 \ldots
u^h_m$. In particular, $m$ could be more than $k$.
A switching vector $\mathbb{V}^h$ of $A$  
for word $w$ is an element of $(L,\calI,L)^m $, where $\mathbb{V}^h[l]=(q,I_l,q')$  
if in the run of $A$ over $w^h$ we have $q \step{u^h_l} q'$.

Let $w'^h = u^h_1 \# u^h_{2}\# \ldots u^h_{m}\#$, where  
$u^h_i=(a^h_{i1},t^h_{i1}), (a^h_{i2},t^h_{i2}) \ldots
(a^h_{i,s_i},t^h_{i,s_i})$  
is a stack-$h$ context, where $s_i=|u^h_i|$. 
Now we assign time stamps of the last letter read in the previous
contexts to the current symbol $\#$ to get the word
%$\kappa^h=u^h_1 (\#,t^h_{1,last}) u^h_{2}(\#,t^h_{2,last}) \ldots u^h_{m}(\#,t^h_{m,last})$.
$\kappa^h=u^h_1 (\#,t^h_{1,s_1}) u^h_{2}(\#,t^h_{2,s_2}) \ldots
u^h_{m}(\#,t^h_{m,s_m})$.
 
We take the word $w'^h$ and looking at this word we construct another
word $\bar{w}^{h}$ by inserting symbols $\# '$ at places where the stack is empty
after popping some symbol, and if $\#'$ is immediately followed by
$\#$ then we drop $\#$ symbol. We do this in a very canonical way as follows:
In this word $w'^h$ look at the first call position $c_1$ and  
its corresponding return position $r_1$.  
Then we insert $\# '$ after position $r_1$ in $w^h$. 
Now we look for next call position $c_2$ and its corresponding return
position $r_2$ and insert symbol $\# '$ after $r_2$. We repeat this
construction for all call and its corresponding return positions in
$w'^h$ to get a timed word $\barw^{h}$ over $\Sigma \cup \{\#, \#'\}$. 
Let $\barw^{h}=\baru^h_1 \widehat{\#} \baru^h_2 \widehat{\#} \ldots
\widehat{\#} \baru^h_z$, where $\widehat{\#}$ is either $\#$ or $\#'$, and 
$\baru^h_i=(\bara^h_{i1},\bart^h_{i1}), (\bara^h_{i2},\bart^h_{i2})
\ldots (\bara^h_{i,s_i},\bart^h_{i,s_i})$, is a timed word.  
 
The restriction of $A$ which reads $\barw^h$ is denoted by $A^h_k$.  
Assign timestamps of the last letter read in the previous
contexts to the current symbol $\hat{\#}$ to get the word
$\bar{\kappa}^h=
\baru^h_1 (\hat{\#},\bart^h_{1,s_1})
\baru^h_{2}(\hat{\#},\bart^h_{2,s_2}) \ldots
\baru^h_{z}(\hat{\#},\bart^h_{z,s_z})$, where $s_i=|\baru^h_{i}|$ for
$i$ in $\{1,\ldots,z\}$. 

A stack-$h$ \emph{switching vector} $\barbbv^h$ is a $z$-tuple of the form $(L, \calI, L)^z$,  
where $z > 0$ and for every $j \leq z$ if $\barbbv^h[j] = (q_j, I_j, q'_j)$ then there is a run of  
$A^h$ from location $q_j$ to $q'_j$.  
 
By definition of $k$-scoped  
word we are guaranteed to find maximum $k$ number of $\#$ symbols from $c_j$ to $r_j$. 
And we also know that stack-$h$ is empty whenever we encounter $\#'$ in the word.  
In other words, if we look at the switching vector $\barbbv^h$ of $A$
reading $\barw^h$, it can be seen as a product of switching vectors of
$A$ each having a length less than $k$. 
Therefore, $\barbbv^h = \prodover_{i=1}^r V^h_i$ where $r \leq z$ and  
$V^h_i = (L \times \calI \times  L )^{\leq k}$.  
When we look at a timed word and refer to the switching vector corresponding
to it, we view it as tuples of switching pairs, but when we look at
the switching vectors as a part of state of $A^h_k$ then we see
%the switching vectors as a part of state of $A^h$ or $A^h_k$ then we see
at a product of switching vectors of length less than $k$.

A \emph{correct sequence of context switches}  for $A^h_k$ wrt $\barkappa^h$
is a sequence of pairs  \ram{$\barbbv^h=P^h_1 P^h_2 \dots P^h_z$}, where 
$P^h_{i}=(\ell^h_i,I^h_i, \ell'^h_i)$, \ram{$2 \leq h \leq n$, }
$P^h_{1}=(\ell^h_{1},\nu^h_{1}, \ell'^h_{1})$ 
and  $I^h_i \in \mathcal{I}$ such that 
%\cmnt{Fix this: Type of $P^h_{1}$ different from given in a switching vector} 
\begin{enumerate}
\item Starting in $\ell^h_{1}$, with the $h$-th stack containing  
$\bot^h$, and an initial valuation $\nu^h_{1}$ of all recorders and  
predictors of $\Sigma^h$, the \ram{\dtecmvpa{}} $A$ processes $u^h_{1}$ and  
reaches some $\ell'^h_{1}$ with stack content $\sigma^h_{2}$ 
and clock valuation $\nu'^h_{1}$. The processing of $u^h_{2}$ by $A$ 
 then starts at location $\ell^h_{2}$, and a time $t \in I^h_{2}$
has elapsed between the processing of $u^h_{1}$ and $u^h_{2}$.  Thus, $A$ 
 starts processing $u^h_{2}$ in $(\ell^h_{2}, \nu^h_{2})$ where 
$\nu^h_{2}$ is the valuation of all recorders and predictors updated  
from $\nu'^h_{1}$ with respect to $t$. The stack content remains 
same as $\sigma^h_{2}$ when the processing of $u^h_{2}$ begins.

\item In general, starting in $(\ell^h_{i}, \nu^h_{i})$, $i >1$ with the  
$h$-th stack containing $\sigma^h_{i}$, and $\nu^h_{i}$ obtained from  
$\nu^h_{i-1}$ by updating all recorders and predictors based on the  
time interval $I^h_{i}$ that records the time elapse between processing 
$u^h_{i-1}$ and $u^h_{i}$, $A$ processes $u^h_{i}$ and reaches  
$(\ell'^h_{i}, \nu'^h_{i})$ with stack content $\sigma^h_{i+1}$.  
The processing of $u^{h}_{i+1}$ starts after time $t \in I^h_{i+1}$ has  
elapsed since processing  $u^h_{i}$ in a location $\ell^h_{i+1}$, and  
stack content being $\sigma^h_{i}$.
\end{enumerate}

These switching vectors were used in 
to get the determinizability of $k$-round \ram{$\dtecmvpa$}~\cite{BDKRT15}
In a $k$-round \dtecmvpa{}, we know that there at most $k$-contexts of stack-$h$ and 
hence the length of switching vector (whichever it is) is at most $k$ for
any given word $w$. See for example the MVPA corresponding to Kleene
star of language given in the Example \ref{lab:ex:lan1}. 
In $k$-scope MVPA for a given $w$, we do not know beforehand what is the
length of switching vector. So we employ not just one switching vector
but many one after another for given word $w$, and we maintain that
length of each switching vector is at most $k$. This is possible because
of the definition of $k$-scope $\dtecmvpa{}$ and Lemma~\ref{lab:lm:kscope-spilt}.

\begin{lemma}(Switching Lemma for $A^h_k$)
\label{switch}
Let $A=(k, L, \AlphabetSet, \Gamma, L^0, F, \Delta)$ be a
\ram{$k$-scope-\dtecmvpa{}.}
%Let $w$ be a timed word accepted by $A$.  %with $u_i=u_{i1}u_{i2} \dots u_{in}$, $1 \leq i \leq k$.
Let $w$ be a timed word \ram{with $m$ maximal $h$-contexts and} accepted by $A$ .  %with $u_i=u_{i1}u_{i2} \dots u_{in}$, $1 \leq i \leq k$.
Then we can construct a \ram{\dtecvpa{}} $A^h_k$ over $\Sigma^h \cup \{\#,\#'\}$ 
%$A^h_k$ has a run over $\barw^{h}$ witnessed by a switching sequence
such that $A^h_k$ has a run over $\barw^{h}$ witnessed by a switching sequence
%$\barbbv^h = \prodover_{i=1}^r \barbbv^h_i$ where $r \leq m$ and
$\barbbv^h = \prodover_{i=1}^r \barbbv^h_i$ where \ram{$r \leq z$}
and $\barbbv^h_i = (L \times \calI \times  L )^{\leq k}$ which ends in the last component 
$\barbbv^h_r$ of $\barbbv^h$ \Iff there exists a $k$-scoped switching sequence
    $\barbbv'^h$ of
    %switching vectors of $A$ such that for any $v'$ of $\barbbv'$ there exist
    switching vectors of $A$ such that for any $v'$ of
\ram{$\barbbv'^h$} there exist
    $v_i$ and $v_j$ in $\barbbv'$ with $i \leq j$ and $v'[1]=v_i[1]$
    %and $v'[last]=v_j[last]$. 
    \ram{and $v'[|v'|]=v_j[|v_j|]$. }
\end{lemma}
\begin{proof}        
 
    %We construct a \ram{\dtecvpa{}} $A^k_h$ as follows :
    We construct a \dtecvpa{} %$A^k_h$ as follows :
    $A^k_h= (L^h, \AlphabetSet \cup \{\#,\#'\}, \Gamma^h, L^0, F^h=F, \Delta^h)$
%as follows:
where, 
%    $L^h  \subseteq (L  \times \mathcal{I} \times L )^m \times
%    \Sigma \cup \{\#,\#'\}$ where $m \leq k$ and $\Delta^h$ is given below. 
$L^h  \subseteq (L  \times \mathcal{I} \times L )^{\leq k} \times
    \Sigma \cup \{\#,\#'\}$ and $\Delta^h$ are given below. 
\ram{\begin{enumerate}  
\item For $a$ in $\Sigma$:\\
    $(P^h_{1}, \ldots , P^h_{i}=(q,I^h_i,q'),b) \step{a, \phi} (P^h_{1},
	\ldots , P'^h_{i}=(q,I'^h_i,q''),a)$, when 
	$q' \step{a,\phi} q''$ is in $\Delta$, and $b \in \Sigma$. 
\item For $a$ in $\Sigma$:\\
    $(P^h_{1}, \ldots , P^h_{i}=(q,I^h_i,q'),\#) \step{a, \phi \wedge
	x_{\#}=0} (P^h_{1}, \ldots , P'^h_{i}=(q,I'^h_i,q''),a)$, when 
	$q' \step{a,\phi} q''$ is in $\Delta$, and $b \in \Sigma$. 
\item For $a$ in $\Sigma$:\\
    $(P^h_{1}, \ldots , P^h_{i}=(q,I^h_i,q'),\#') \step{a, \phi \wedge
	x_{\#'}=0} (P^h_{1}, \ldots , P'^h_{i}=(q,I'^h_i,q''),a)$, when 
	$q' \step{a,\phi} q''$ is in $\Delta$, and $b \in \Sigma$. 
\item For $a=\#$, \\
 $(P^h_{1}, \ldots, P^h_{i}=(q,I^h_i,q'),b) \step{a,\phi \wedge x_{b} \in I'^h_{i+1}}  
	(P^h_{1}, \ldots , P'^h_{i+1}=(q'',I'^h_{i+1}, q''),\#)$, when 
$q' \step{a,\phi} q''$ is in $\Delta$. 
\item For $a=\#'$, \\
    $(P^h_{1}, \ldots , P^h_{i}=(q,I^h_i,q'),a) \step{a,\phi,x_{\#'} \in
	\hat{I}^h_1}  
	(\hat{P}^h_{1}=(q',\hat{I}^h_1,q'),\#')$, when 
	$q' \step{a,\phi} q''$ is in $\Delta$. 
\end{enumerate}}

Given a timed word $w$ accepted by $A$, when $A$ is restricted to $A^h$ then it is
    running on $w'^h$, the projection of $w$ on $\Sigma^h$,
\ram{interspersed with $\#$ separating the maximal $h$-contexts in
original} word $w$.  
Let $v_1,v_2, \ldots, v_m$ be the sequence of switching vectors
witnessed by $A^h$ while reading $w'^h$. 
% which is a collection of stack-$h$ contexts.

    Now when $w'^h$ is fed to the constructed machine $A^k_h$, 
    it is interspersed with new symbols $\#'$ whenever the stack is
    empty just after a return symbol is read. 
    \ram{Now $\bar{w}^{h}$ thus constructed is again a collection of
$z$ stack-$h$  }
    %Now $w^{h'}$ thus constructed is again a collection of stack-$h$
    contexts which possibly are more in number than in $w'^h$. 
    And each newly created context is either equal to some context of
    $w'^h$ or is embedded in exactly one context of $w'^h$.  
    These give rise to sequence of switching vectors 
    $v'_1, v'_2, \ldots, v'_z$, where \ram{$m \leq z$}. That explains
    the embedding of switching vectors witnessed by $A^k_h$, while
    reading $\bar{w}^{h}$, into switching vectors of $A$, while
reading \ram{$w^h$. }
\qed
\end{proof} 

Let $w$ be in $L(A)$. Then as described above we can have a sequence of
switching vectors $\barbbv_h$ for stack-$h$ machine $A^h_k$.  
%One can see $\barbbv^h$ as just a sequence of tuples of states of $A$.  
Let $d^h$ be the number of $h$-contexts in the $k$-scoped splitting of
$w$ i.e., the number of $h$-contexts in $\barw^h$.  
%Then we have those many tuples in the sequence of SVs $\barbbv^h$. 
Then we have those many tuples in the sequence of \ram{switching vectors} $\barbbv^h$. 
Therefore, $\barbbv^h = \Pi_{y \in \{1,\ldots,d_h\}}  \langle l^h_y, I^h_y,l'^h_y \rangle$.
 
We define the relation between elements of $\barbbv^h$ across all such sequences.
While reading the word $w$, for all $h$ and $h'$ in $\{1,\ldots,n\}$ and 
for some $y$ in $\{1,\ldots, d_h\}$ and some $y'$ in $\{1,\ldots, d_{h'}\}$ 
we define a relation $\emph{follows}(h,y)=(h'y')$ if $y$-th $h$-context is
followed by $y'$-th $h'$-context.  
%A collection of sequence of switching vectors

\ram{A collection of correct sequences of context switches given via switching vectors}
$(\barbbv^1,\ldots,\barbbv^n)$ is called %\concept{globally correct},  
\concept{globally correct} 
if we can stitch together runs of all $A^h_k$s on $\barw^h$ using these switching vectors 
to get a run of $A$ on word $w$.

In the reverse direction, if for a given $k$-scoped word $w$ over
$\Sigma$ which is in $L(A)$ then we have,  collection of globally
correct switching vectors $(\barbbv^1,\ldots,\barbbv^n)$.

The following lemma enables us to construct a run of $k$-scope
\ecmvpa{} on word $w$ over $\Sigma$ from the runs of \ecvpa{} $A^j_k$
on $h$-projection of $w$ over $\Sigma^j$ with the help of switching vectors. 
 
%The detailed proof of the following lemma is given in Appendix~\ref{lab:app:stitch}. 

% 
%
%\begin{lemma}[Stitching Lemma]
\begin{lemma}[Stitching Lemma]
%Let $A=(k, L, \AlphabetSet, \Gamma, L^0, F, \Delta)$ be a $k$-scope-\ECMVPA{}.
Let $A=(k, L, \AlphabetSet, \Gamma, L^0, F, \Delta)$ be a $k$-scope
\dtecmvpa{}.
%\ram{\dtecmvpa{}.}
Let $w$ be a $k$-scoped word over $\Sigma$.  %with $u_i=u_{i1}u_{i2} \dots u_{in}$, $1 \leq i \leq k$.
    Then $w \in L(A)$ iff there exist a collection of globally correct
    sequences of switching vectors for word $w$.
    \label{lab:lm:stitching}
\end{lemma}  
\begin{proof}        
$(\Rightarrow)$: Assuming $w \in L(A)$ to prove existence of a collection
of globally correct sequences is easy.  
 
$(\Leftarrow)$:  Assume that we have a collection of globally
correct switching vectors $(\barbbv^1,\ldots,\barbbv^n)$ of $A$, for a
word $w$. 
    For each $h$ in $\{1,\ldots,n\}$  we have $k$-scoped splitting of word
$w$ and we have a run of $A^h_k$ on $\barw^h$, which uses
$\barbbv^h$, the corresponding switching vector of $A$.  
 
    Let $\barw$ be the word over $\Sigma \cup \{\#^1,\ldots,\#^n\} \cup 
    %\{\#'^1,\ldots, \#'^h\}$, obtained from $w$ by using
    \ram{\{\#'^1,\ldots, \#'^n\}}$, obtained from $w$ by using
    $k$-splitting as follows. 
    Let $w= w_1 w_2 \ldots w_d$
    We first insert $\#^j$ before $w_1$ if $w_1$ is a $j$-context. Then 
    for all $i$ in $\{1,\ldots,d\}$, insert $\#^l$ in between $w_i$ and
    \ram{$w_{i+1}$} if $w_{i+1}$ is $l$-th context. Let us insert special symbol
    $\#^{\$}$ to mark the end of word $w$. 
    Now using $k$-splitting of word $w$ we throw in $\#'^h$ as done in the
    case of obtaining $\barw^h$ from $w^h$. 

%Now we build a composite machine whose constituents are $A^h_k$ for all 
%$h$ in $\{1,\ldots,n\}$. Its initial state is  
%$(p^1,\ldots, p^j,\ldots, p^n,j)$ where $p_j$ s are the initial state
%of $A^h_k$s and the last component tells which component is processing
%current symbol. According to this we are processing $j$-th context
%$\barw^h$.  
% 
\ram{Now we build a composite machine whose constituents are $A^h_k$ for all 
$h$ in $\{1,\ldots,n\}$. Its initial state is  
$(p^1,\ldots, p^j,\ldots, p^n,j)$ where $p^j$ s are the initial states
of $A^j_k$s and the last component tells which component is processing
current symbol. According to this, initially we are processing first $j$-context
in $\barw^j$.}
 
We first run $A^j_k$ on $\barw^j$ updating location $p^j$ to $p'^j$ where  
$p'^j= (\barv^j=(l^j_1,I^j_1,l'^j_1),a)$. When it reads $\#^g$ then 
$p'^j= (\barv^j=(l^j_1,I^j_1,l'^j_1),\#^g)$ and the composite
state is changed to $(p^1,\ldots, p'^j,\ldots, p^n,g)$  meaning that  
next context belongs to $g$-th stack.  
%So we start simulating $A^j_k$ on $g$-context in $\barw^g$.  
 \ram{So we start simulating $A$ on first $g$-context in $\barw^g$}.  
But to do that $A^j_k$ should have left us in a state where
$A^g_k$ can begin. That is if $p^g= (\barv^g=(l^g_1,I^g_1,l'^g_1),\#^g)$   
    we should have $l^g_1=l'^j_1$ and $x_{\#'^g}=0$ which will be the
    case as this is the first time we accessing $g$-th stack. 
 
%In general, if we are leaving a $j$-th context and the next context
In general, if we are leaving a \ram{$j$-context} and the next context
belongs to $g$-th stack then the time elapsed from reading the last symbol of  
%last $g$-th context should fall in the interval in the last element of  
last \ram{$g$-context} should falls in the interval of the  
%last element of  
\ram{first} element of next switching vector processing \ram{next
$g$-context} in  $\barw^g$, along with the matching of  
state in which the previous context has left us in. 

\end{proof}  
%\qed

\section{Emptiness checking and Determinizability of scope-bounded ECMVPA}
\label{lab:sec:det-kecmvpa}

First we show that emptiess problem is decidable using the ideas from
\cite{BDKPT16}. Fix a $k \in \Nat$. 
 
%\subsection{Emptiness checking for $k$-scope \ecmvpa}
 
\begin{theorem}      
\label{lab:tm:empt-kecmvpa}
    Emptiness checking for $k$-scope \ecmvpa{} is decidable.
\end{theorem} 
\begin{proof}[Proof sketch]
    Decidability of emptiness checking of $k$-round \ecmvpa{} has been shown
    in \cite{BDKPT16}. This proof works for any general \ecmvpa{} as the
notion $k$-round has not been used. 
So, the same proof can be used to decide emptiness checking of
    $k$-scope \ecmvpa{}.  
\end{proof} 
 
%\subsection{Determinization of $k$-scope \ecmvpa{}}

Rest of the section is devoted for the proof of following theorem.  
 
\begin{theorem}      
\label{lab:tm:det-kecmvpa}
    The class of $k$-scope \ecmvpa{} are determinizable. 
\end{theorem} 
 
To show this we use the determinization of VPA \cite{AM04vpda}
and we recall this construction here for the reader's convenience.
\subsubsection*{Determinization of VPA \cite{AM04vpda}}
\label{vpa-det}
 
Given a VPA $M=(Q, Q_{in}, \Gamma, \delta, Q_F)$, the idea in \cite{AM04vpda} is to do a subset construction. 
Let $w=w_1a_1w_2a_2w_3$ be a word such that every call in $w_1, w_2, w_3$ has a matching return, and 
$a_1, a_2, a_3$ are call symbols without matching returns. After reading $w$, the deterministic VPA has  
stack contents $(S_2,R_2,a_2)(S_1,R_1,a_1)\bot$ and is in control state $(S,R)$.  
Here, $S_2$ contains all pairs of states $(q,q')$ such that starting with $q$ on $w_2$ and  
an empty stack (contains only $\bot$), we reach $q'$ with stack
$\bot$. The set of pairs of states $S_2$ is called a summary for $w_2$. Likewise, $S_1$ is a summary for $w_1$ 
and $S$ is the summary for $w_3$. Here $R_i$ is the set of states reachable from the initial state 
after reading till the end of $w_i$, $i=1,2$ and $R$ is the set of reachable states obtained on reading $w$. 

After $w_3$, if a call $a_3$ occurs, then $(S,R,a_3)$ is pushed on the stack, and the current state 
is $(S',R')$ where $S'=\{(q,q)\mid q \in Q\}$, while $R'$ is obtained by updating  $R$  using all transitions  
for $a_3$. 
The current control state $(S,R)$ is updated to $(S', R')$ where $R'$
in the case of call and internal symbols is the set of all reachable states obtained 
from $R$, using all possible transitions on the current symbol read,
and where the set $S'$ is obtained as follows:
\begin{itemize}
\item On reading an internal symbol, $S$ evolves into $S'$ where $S'=\{(q,q') \mid \exists q'', (q, q'') \in S, (q'', a, q') \in \delta\}$.
\item On reading a call symbol $a$, $(S, R,a)$ is pushed onto the stack, and the control state is $(S', R')$ where 
$S'=\{(q,q) \mid q \in Q\}$. On each call, $S'$ is re-initialized. 
 \item On reading a return symbol $a'$, let the top of stack be $(S_1, R_1, a)$. This is popped. 
 Thus, $a$ and $a'$ are a matching call-return pair. Let 
the string read so far be $waw'a'$. Clearly, $w, w'$ are well-nested, or all calls in them have seen their returns.   
 
For the well-nested string $w$ preceding $a$, we have  $S_1$ consisting of  
all $(q,q'')$ such that starting on $q$ on $w$, we reach $q''$ with empty stack.   
Also, $S$ consists of pairs $(q_1,q_2)$ that have been obtained since the call  
symbol $a$ (corresponding to the return symbol $a'$) was pushed onto the stack.  
The set $S$ started out as $\{(q_1,q_1) \mid q_1 \in Q\}$ on pushing $a$, and contains  
pairs $(q_1,q_2)$ such that on reading the well-nested string between $a$ and $a'$,  
starting in $q_1$, we reach $q_2$. The set $S$ is updated to $S'$ by ``stitching''  
$S_1$ and $S$ as follows:  
 
A pair $(q,q') \in S'$ if  there is some  $(q, q'') \in S_1$, 
and $(q'',a,q_1,\gamma) \in \delta$ (the push transition  
on $a$), $(q_1, q_2) \in S$, and $(q_2, a', \gamma, q')\in \delta$ (the pop transition on $a'$).  

In this case, a state $q' \in R'$ if there exists some location $q$ in
$R_1$  and there exists and $(q,a,q_1,\gamma) \in \delta$ (the push transition  
on $a$), $(q_1, q_2) \in S'_1$, and $(q_2, a', \gamma, q')\in \delta$
(the pop transition on $a'$).  The important thing is the reachable
set of states of non-deterministic machine after the call transition is  
updated using all possible summaries of well-matched words possible after  
the corresponding push transition in the past. 
\end{itemize}
 
The set of final locations of the determinized VPA are $\{(S,R) \mid R$  
contains a final state of the starting VPA$\}$, and its initial location  
is the set of all pairs $(S_{in}, R_{in})$ where $S_{in}=\{(q,q) \mid q \in Q\}$  
and $R_{in}$ is the set of all initial states of the starting VPA.

%\subsection{Proof of Theorem \ref{lab:tm:ecmvpa-det}} 
\subsubsection*{Determinization of $k$-scope ECMVPA} 
In this section we show that $k$-scope \ecmvpa{} are determinizable using
the result from \cite{VTO09} that single stack \ecvpa{} are
determinizable.  

Let $A=(k, L, \AlphabetSet, \Gamma, L^0, F, \Delta)$ be the $k$-scoped \ecmvpa{} and
let $A^h_k$ be the \ECVPA{} on $\Sigma^h \cup \{\#^h, \#'^h\}$. Each $A^h_k$ is  
determinizable \cite{VTO09}. Recall from~\cite{VTO09} that an \ECVPA{} $A^h_k$  
is untimed to obtain a \VPA{} $ut(A^h_k)$ by encoding the clock constraints of  
$A^h_k$ in an extended alphabet. This \VPA{} can be converted back into an  
\ECVPA{} $ec(ut(A^h_k))$ by using the original alphabet, and replacing the  
clock constraints. This construction is such that $L(ec(ut(A^h_k)))=L(A^h_k)$  
and both steps involved preserve determinism. Determinization of \VPA{} $ut(A^h_k)$  
is done in the usual way \cite{AM04vpda}. This gives $Det(ut(A^h_k))$.  
Again, $ec(Det(ut(A^h_k)))$ converts this back into a \ECVPA{} by simplifying the  
alphabet, and writing the clock constraints. The set of locations remain unchanged  
in $ec(Det(ut(A^h_k)))$ and $Det(ut(A^h_k))$. This translation also preserves determinism,  
hence $B^h_k=ec(Det(ut(A^h_k)))$ is a deterministic \ECVPA{} language equivalent  
to \ECVPA{} $A^h_k$. 
This process is also explained in the Figure~\ref{fig:proof-dia}.

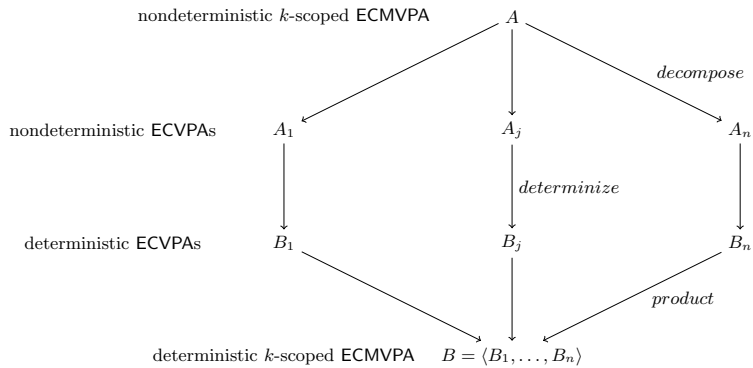
\begin{figure} 
\begin{center}\scalebox{0.75}{
      \begin{tikzpicture}
 
    \tikzstyle{every node}=[draw=none,fill=none];

%% first round

    \node (A0) at (0,4) {nondeterministic $k$-scoped \ecmvpa};
    \node (B0) at (0,-2) {deterministic $k$-scoped \ecmvpa};
	  \node (A) at (4,4) {$A$};
    \node (A1) at (0,2) {$A_1$};
    \node (A10) at (-3,2) {nondeterministic \ecvpa s};
    \node (Aj) at (4,2) {$A_j$};
    \node (An) at (8,2) {$A_n$};

    \node (B1) at (0,0) {$B_1$};
    \node (B10) at (-3,0) {deterministic \ecvpa s};
    \node (Bj) at (4,0) {$B_j$};
    \node (Bn) at (8,0) {$B_n$};
     
    \node (B) at (4,-2) {$B=\langle B_1,\ldots,B_n\rangle$};

   \draw[->,black,rounded corners] (A)--(A1); 
   \draw[->,black,rounded corners] (A)--(Aj); 
   \draw[->,black,rounded corners] (A)--(An);

   \draw[->,black,rounded corners] (A1)--(B1); 
   \draw[->,black,rounded corners] (Aj)--(Bj); 
   \draw[->,black,rounded corners] (An)--(Bn);

   \draw[->,black,rounded corners] (B1)--(B); 
   \draw[->,black,rounded corners] (Bj)--(B); 
   \draw[->,black,rounded corners] (Bn)--(B);

    \tikzstyle{every node}=[draw=none,fill=none];
   \node (decompose) at (7.3,3) {$decompose$};
   \node (determinize) at (5,1) {$determinize$};
   \node (product) at (7,-1) {$product$};
\end{tikzpicture}
} 
\end{center} 
\caption{Determinazation of $k$-scoped \ecmvpa}
\label{fig:proof-dia}
\end{figure}

The locations of $B^h_k$ are thus of the form 
$(S,R)$ where $R$ is the set of all reachable control states of %$ut(A_j)$
$A^h_k$ and $S$ is a set of ordered pairs of states of $A^h_k$ as seen in section \ref{vpa-det}. 
On reading $\kappa_j$, the $R$ component of the state reached in
$B^h_k$ is the set  $\{\langle V^h\rangle \mid V^h ~\mbox{is a last
component of switching sequence}~ \barbbv^h$ of $A^h_k\}$. 
Lemmas \ref{switch} and Lemma \ref{lab:lm:stitching} follow easily  using 
  $B^h_k=ec(Det(ut(A^h_k)))$ in place of $A^h_k$. 
%\noindent {\bf Construction of deterministic \ECMVPA{} $B$ simulating $B_1, \dots, B_n$}
We now obtain a deterministic $k$-scoped \ecmvpa{} $B$ for language of
$k$-scoped \ECMVPA{}
$A$ by simulating $B^1_k, \dots, B^k_k$ one after the other on reading
$w$, with the help of globally correct sequences of k-scoped switching
vectors of $A^h_k$ s. 

Automaton $B$ in its finite control, keeps track of the locations of all
the $B^h_k$'s, along with the valuations of all the recorders and predictors of $\Sigma$.   
It also remembers the current context number of the $B^h_k$'s to
ensure correct simulation.

Let $B^h_k=(Q^h, \Sigma^h \cup\{\#,\#'\},\Gamma^h, Q^h_0, F^h, \delta^h)$. 
Locations of $B^h_k$ have the form $(S^h, R^h)$. The initial state of $B^h_k$  
is the set consisting of all $(S_{in}, R_{in})$ where $S_{in}=\{(q,q) \mid q$  
is a state of $A^h_k\}$, and $R_{in}$ is the set of all initial states
of $A^h_k$.  
Recall that a final state of $A^h_k$ is $V^h$ the last component
switching vector $V^h$ of a correct switching sequence  
$\barbbv^h$ of $A^h_k$. Thus, an accepting run in $B^h_k$ goes through states 
$(S_{in}, R_{in}), (S_1, R_1) \dots, (S_m, R_m),  \langle V^h
\rangle$.  

Locations of $B$ have the form $(q_1, \dots, q_n, h)$ where $q_y$ is a  
location of $B_y$, $h$ is the stack number of current context.  
Without loss of generality we assume that the first context in any
word belongs to stack $1$. The initial location of $B$ is  
$(q_{11}, q_{12}, \dots, q_{1n},1)$ where $q_{1h}$ is the initial  
location of $B^h_k$. We define the set of final locations of $B$ to be  
$(\langle V^1 \rangle, \ldots \langle V^{n-1} \rangle \langle S_n,R_n \rangle)$  
where $R_n$ is  a set containing a tuple of the form $(k,l'_{kn}, V_n,a)$ and
$ l'_{kn}$ is in $F$, the set of final locations of $A$.

We now explain the transitions $\Delta$ of $B$, using the transitions
$\delta^j$ of $B^h_k$.  
Recall that $B$ processes $w=w_1w_2 \dots w_m$ where each $w_i$ is
context of some stack.%, with  $u_i = u_{i1}u_{i2} ...u_{in}$.
When $w$ is $k$-scoped we can see $w$ as the concatenation of contexts 
$w=\barw_1\barw_2\ldots\barw_{\barm}$
where consecutive contexts need not belong to different stacks, 
and $\barm \geq m$. 
Let $\eta=(q_{1}, \dots,  q_{h-1}, q_{h}, q_{h+1}, \dots, q_{n},h)$
and let \\
$\zeta= (q_{1}, \dots,  q_{h-1}, q, q_{h+1}, \dots, q_{n},h)$,  
where $q_{h}$ is some state of $B^h_k$ while processing some context
of $h$-th stack.

\begin{enumerate}
\item Simulation of $B^h_k$ when the  $\barw_{\barm}$ 
is not an $h$-context. 

\begin{itemize}
\item $\langle \eta, a, \varphi, \zeta \rangle \in\Delta^j_{l}$ iff
$(q_{ij},a, \varphi, q) \in \delta_{l}^j$
 \item $\langle \eta, a, \varphi, \zeta, \gamma \rangle \in \Delta^j_c$ iff  $(q_{ij},a, \varphi,\gamma,  q) \in \delta_{c}^j$
 \item $\langle \eta, a, \gamma, \varphi, \zeta \rangle \in \Delta^j_r$ iff  $(q_{ij},a, \gamma, \varphi, q) \in \delta_{r}^j$
\end{itemize}
\item Change context from $h$ to $h+1$. 
Assume that we have processed $h$-context $\barw_g$ and next context is
$\barw_{g+1}$ which is an $(h+1)$-context.  
Let the current state of $B$ is   
$\eta=(q_{1}, \dots,  q_{h-1}, q_{h}, q_{h+1}, \dots, q_{n},h)$ 
after processing $\barw_g$. 
At this point $B^h_k$ reads symbol $\#$ in $\barkappa^h$ and moves from state $q_{h}$ to $q'_{h}$. 
Therefore, current state of $B$ becomes 
$\eta'=(q_{1}, \dots,  q_{h-1}, q'_{h}, q_{h+1}, \dots, q_{n},h)$ 
	At this point, $B$ invokes $B^{h+1}_k$ and starts reading first
	symbol of $\barw_{g+1}$ when $B$ is in state $\eta'$. 
	This is meaningful because of Lemma~\ref{lab:lm:stitching}. States of
	$q_{h+1}$ are of the form $(S^{h+1}, R^{h+1})$ where 
        $R^{h+1}$ have states of $A^{h+1}_k$ of the form
	$(\barbbv^{h+1},a)$, 
	which is a possible state after processing last $h+1$ context. 
        Globally correct stitching sequence guarantees that for all 
	$(\barbbv^h,a) \in R'^h$ we have at least one
	$(\barbbv^{h+1},a) \in R'^{h+1}$ such that 
	$\barbbv'^h[last]=(x,I,y)$ then 
	$\barbbv'^{h+1}[first]=(y,I',z)$ and valuation of clocks at
        $\eta'$ belongs to $I'$, where $x,y$ are locations of $A$, and
\ram{where last is the last component of $\barbbv'^h$ after processing
$\barw_g$, and first is the first componet of switching vector
$\barbbv'^{h+1}$ for processing $\barbbv'{h+1}$.}
Component location $q_{h+1}$ will be replaced based on a transition 
of $B^{h+1}_k$, and $q'^{h}$ is also replaced with $q^{h}$ to take care of the transition 
on $\#$ in $B^h_k$, where $q^{h}=(S^{h}, R^{h})$ with $R^{h}$ containing all locations  
of the form $(\barbbv^h, a)$, where $a$ is the last symbol of
	$\barw_{g}$.  

\begin{itemize}
\item $\langle (\dots,q'_{h},q_{h+1}, \dots,h), a, \varphi,
(\dots,q_{h},q, \dots,h+1) \rangle \in \Delta^{h+1}_{l}$ iff \\
$(q_{h+1}, a, \varphi, q) \in \delta^{h+1}_{l}$
\item $\langle (\dots,q'_{h},q_{h+1}, \dots,h), a, \varphi,  (\dots,q_{h},q, \dots,h+1), \gamma \rangle \in \Delta^{h+1}_{c}$ iff \\
$(q_{h+1}, a, \varphi, q, \gamma) \in \delta^{h+1}_{c}$
\item $\langle (\dots,q'_{h},q_{h+1}, \dots,h), a, \gamma, \varphi,  (\dots,q_{h},q, \dots,h+1) \rangle \in \Delta^{h+1}_{r}$ iff \\
$(q_{h+1}, a, \gamma, \varphi, q) \in \delta^{h+1}_{r}$
\end{itemize}
Transitions of $B^{h+1}$ continue on 
$(\dots,q_{h},q, \dots,h+1)$ replacing only the $(h+1)$-th entry 
until $\barw_{g+1}$ is read completely.

\item Reading $\barw_{\barm}$ the last context of word $w$. Without loss of  
    generality assume that this is the context of $B^n_k$ and the previous  
    context $\barw_{\barm-1}$ was the context of $B^{n-1}_k$. 
When $B^{n-1}_k$ stops reading last symbol of $\barw_{\barm}$ then  
it is in the state 
$\eta=(\mathcal{V}^1,
\mathcal{V}^2,\ldots,\mathcal{V}^{n-1},(S'^n,R'^n),n)$
where no more symbols are to be read, and each $B^h_k$ is in the state
$\langle \mathcal{V}^h \rangle$, where $\mathcal{V}^h$ is the last
component of $k$-scoped switching vector $\barbbv^h$,
where $R'_{n}$ is the set of all locations of the form 
 $(\mathcal{V}^n,a)$, where $a$ is the last symbol of $\barbbv^h$. 
    This is accepting iff there exists $\ell'_{kn}$ as the second
    destination of last tuple in $\mathcal{V}^n$  and $\ell'_{kn} \in F$. 
Note that we have ensured the following:
\begin{enumerate}
\item  $\barbbv^h= \prodover_{1 \leq i \leq  m_h} \mathcal{V}^h_i$ and is part of correct  
global sequence for all $1 \leq h \leq n$, and $m_h$ is the number of $k$-switching  
vectors of $A^h_k$ used in the computation.   
 
\item At the end of $\barw_{arm}$, we reach in $B^h_k$, $(S'_{kn},R'_{kn})$  such that 
 there exists $\ell'_{kn}$ as the second
 destination of the last tuple in $\mathcal{V}^n$  and $\ell'_{kn} \in F$. 
\item When switching from one context to another continuity of state
    is used as given in the correct global sequence.  
\end{enumerate}
      \end{enumerate}
 The above conditions ensure correctness of local switching and a globally correct sequence in $A$.  
 Clearly, $w \in L(B)$ iff $w \in L(A)$ iff there is some globally correct sequence 
 $\barbbv^1 \dots \barbbv^n$. 

%This completes the proof of Theorem~\ref{lab:tm:det-kecmvpa}.

\section{Emptiness checking and Determinizability of scope-bounded \dtecmvpa{}}
\label{sec2}
Fix a $k \in \Nat$.
\ram{The proof is via untiming stack construction to get $k$-scope $\ecmvpa$ for
which emptiness is shown to be decidable in Theorem~\ref{lab:tm:empt-kecmvpa}}.
 
 We first informally describe the \emph{untiming-the-stack} construction to obtain
from a \ram{$k$-scope-\dtecmvpa{} $M$}
over $\Sigma$, an \ram{$k$-scope-\ECMVPA{}  $M'$} over an extended alphabet $\Sigma'$ such that $L(M)=h(L(M'))$ where $h$ is a homomorphism $h: \Sigma' {\times} \mathbb{R}^{\geq 0}\rightarrow \Sigma {\times} \mathbb{R}^{\geq 0}$
defined as $h(a,t)=(a,t)$ for $a {\in} \Sigma$ and $h(a,t)=\epsilon$ for $a {\notin} \Sigma$. 
 
Our construction builds upon that of~\cite{BDKRT15}. 

Let $\kappa$ be the maximum constant used in the
\ram{$k$-scope-\dtecmvpa{} $M$} while checking the age of a popped symbol
in any of the stacks. 
Let us first consider a call transition 
$(l, a, \varphi, l', \gamma) \in \Delta^i_c$ encountered in $M$.
To construct an \ram{$k$-scope-\ECMVPA{} $M'$} from $M$, we guess the interval used in the 
return transition when $\gamma$ is popped from $i$th stack. Assume the
guess is an interval of the form  
$[0,\kappa]$. This amounts to checking that the age of $\gamma$ at the time of popping is ${<}\kappa$.
In $M'$, the control switches from $l$ to a special location $(l'_{a, {<}\kappa}, \{{{<_i}\kappa}\})$, 
and the symbol $(\gamma, {{<}\kappa}, \first)$\footnote{It is sufficient to push $(\gamma, {{<}\kappa}, \first)$ 
 in stack $i$, since the stack number is known as $i$}
 is pushed onto the $i$th stack. 

Let $Z_i^{\sim}=\{\sim_i c \mid c {\in} \mathbb{N}, c \leq k, \sim {\in} \{{<}, \leq, >, \geq\}\}$. 
Let $\Sigma'_i=\Sigma^i \cup Z_i^{\sim}$ be the extended alphabet for transitions on the $i$th stack.
All symbols of $Z_i^{\sim}$ are internal symbols in $M'$ i.e.
%$\Sigma'_i=\set{\Sigma^i_c, \Sigma^i_{int}\cup  Z_i^{\sim}, \Sigma^i_r}$. 
$\Sigma'_i=\set{\Sigma^i_c, \Sigma^i_{l}\cup  Z_i^{\sim}, \Sigma^i_r}$. 
At location $(l'_{a,{<}\kappa}, \{{{<_i}\kappa}\})$,  the new symbol ${{<_i}\kappa}$
is read and we have the following transition :
$((l'_{a,{<}\kappa}, \{{{<_i}\kappa}\}), {{<_i}\kappa}, x_a=0,
(l', \{{{<_i}\kappa}\}))$, which  results in resetting the event recorder
$x_{{<_i}\kappa}$ corresponding to the new symbol ${{<_i}\kappa}$.  
The constraint $x_a=0$ ensures that no time is elapsed by the new transition. 
The information ${{<_i}\kappa}$ is retained in the control state until
$(\gamma, {{<}\kappa}, \first)$  is popped from $i$th stack. 
At $(l', \{{{<_i}\kappa}\}))$, we continue the simulation of $M$ from $l'$. 
Assume that we have another push operation on $i$th stack at $l'$ of the form 
$(l',b, \psi, q, \beta)$. In $M'$, from $(l',\{{{<_i} \kappa}\})$, we first guess
the constraint that will be checked when $\beta$ will be popped from the $i$th stack. 
If the guessed constraint is again ${<_i}\kappa$, then control switches from $(l',\{{{<_i} \kappa}\})$ 
to $(q,\{{{<_i} \kappa}\})$, and $(\beta,{{<}\kappa},-)$ is pushed onto the $i$th stack
and simulation continues from  $(q,\{{{<_i} \kappa}\})$. However, if the 
guessed pop constraint is ${<_i}\zeta$ for $\zeta \neq \kappa$, then
 control switches from $(l',\{{{<_i} \kappa}\})$ to $(q_{b,{<}\zeta},\{{{<_i}\kappa},{{<_i} \zeta}\})$ on reading $b$. 
 The new obligation ${<_i}\zeta$ is also remembered in the control state. 
 From $(q_{b,{<}\zeta},\{{{<_i}\kappa},{{<_i} \zeta}\})$, we read the new symbol ${<_i}\zeta$ which resets the event 
 recorder $x_{{<_i}\zeta}$ and control switches to $(q, \{{{<_i}\kappa},{{<_i} \zeta}\})$, pushing 
 $(\beta, {{<}\zeta}, \first)$ on to the $i$th stack. The idea thus is to keep the obligation ${{<_i}\kappa}$ alive 
 in the control state until $\gamma$ is popped; 
 the value of $x_{{<_i}\kappa}$ at the time of the pop determines whether
 the pop is successful or not.
 If a further ${{<_i}\kappa}$ constraint is encountered while the obligation ${{<_i}\kappa}$
 is already alive, then we do not reset the event clock $x_{{<_i}\kappa}$. The $x_{{<_i}\kappa}$ is 
 reset only at the next call transition  after $(\gamma, {{<}\kappa},
\first)$ is popped from $i$th stack , when ${{<_i}\kappa}$ is again guessed.   
The case when the guessed popped constraint is of the form ${>_i}\kappa$ is similar.
In this case, each time the guess is made, we reset the event recorder $x_{{>_i}\kappa}$ at the time of the push.
If the age of a symbol pushed later is ${>}\kappa$, so will be the age of a symbol pushed earlier. 
In this case, the obligation ${{>}\kappa}$ is remembered only in the stack and not in the finite control. 
Handling guesses of the form $\geq \zeta \wedge \leq \kappa$ is similar, and we combine the ideas discussed above.  

\Suggestion{Now consider a return transition $(l, a, I, \gamma, \varphi, l')\in \Delta^i_r$ in $M$.} In $M'$, we 
are at some control state $(l,P)$. On reading $a$, we check the top of the $i$th stack symbol in $M'$.
It is of the form $(\gamma, S, \first)$ or 
$(\gamma, S, -)$, where $S$ is a singleton set of the form 
$\{{<}\kappa\}$ or $\{{>}\zeta\}$, or a set of the form 
$\{{<}\kappa, {>}\zeta\}$\footnote{This last case happens when the age checked lies between $\zeta$ and $\kappa$}. Consider the case when the top of the $i$th stack symbol is $(\gamma, \{{<}\kappa, {>}\zeta\}, \first)$. 
In $M'$, on reading $a$, the control switches from $(l,P)$ to $(l',P')$ for $P'=P \backslash \{{<}\kappa\}$ iff
 the guard $\varphi$ evaluates to true,
the interval $I$ is $(\zeta,\kappa)$ (this validates our guess made at the time
of push) and
the value of clock $x_{{<_i}\kappa}$ is ${<} \kappa$, and the value of clock 
$x_{{>_i}\zeta}$ is  ${>}\zeta$.
Note that the third component $\first$ says that there are no symbols
in $i$th stack  
below $(\gamma, \{{<}\kappa, {>}\zeta\}, \first)$ whose pop constraint is
${<}\kappa$.
Hence,  we can remove the obligation ${<_i}\kappa$ from $P$ in the control state. 
If the top of stack symbol was $(\gamma, \{{<}\kappa, {>}\zeta\},-)$, then we know that 
the pop constraint ${<}\kappa$  is still alive for $i$th stack . That is, there is some stack symbol below 
$(\gamma, \{{<}\kappa, {>}\zeta\},-)$ of the form $(\beta, S, \first)$ such that ${<}\kappa {\in} S$.
In this case, we keep $P$ unchanged and control switches to $(l',P)$. 
Processing another $j$th stack  continues exactly  as above; the set $P$
contains $<_i \kappa, \leq_j \eta$, and so on depending on what 
constraints are remembered per stack.  Note that the set $P$ in $(l, P)$ only contains constraints of the form $<_i \kappa$ or $\leq_i \kappa$ for 
 each $i$th stack, since we do not remember $>\zeta$ constraints in the finite control.
%The formal construction is in Appendix \ref{app:untime}.  

 %\end{itemize}

We now give the formal construction.  
%\subsection*{Reduction from \dtMVPA{} to \ECMVPA{}:}
\subsection*{Reduction from \ram{$k$-scope-\dtecmvpa{} to $k$-scope-\ECMVPA{}:}}
%\label{app:untime}

Let $Z^{\sim}=\bigcup_{i=1}^n Z^{\sim}_i$ and let  
$S^{\sim}=\{\sim c \mid c {\in} \mathbb{N}, c \leq \kappa, \sim {\in} \{{<}, \leq, {>}, \geq, =\}\}$.  
%Given $k$-\dtMVPA{} $M=(L, \Sigma, \Gamma, L^0, F, \Delta)$ with max constant $\kappa$  
Given \ram{$k$-scope-\dtecmvpa{}} $M=(L, \Sigma, \Gamma, L^0, F, \Delta)$ with max constant $\kappa$  
used in return transitions of all stacks, we construct  
%$k$-\ECMVPA{} $M'=(L', \Sigma', \Gamma', L'^0, F', \Delta')$ where
\ram{$k$-scope-\ECMVPA{}} $M'=(L', \Sigma', \Gamma', L'^0, F', \Delta')$ where
$L'{=}(L {\times} 2^{Z^{\sim}}) \cup (L_{\Sigma_i {\times} S^{\sim}} {\times}
2^{Z^{\sim}}) \cup (L_{\Sigma_i {\times} S^{\sim}{\times} S^{\sim}} {\times}
2^{Z^{\sim}})$, 
%$\Sigma_i'=(\Sigma^i_c, \Sigma^i_{int} \cup Z_i^{\sim}, \Sigma^i_r)$ and $\Gamma_i'=\Gamma_i
$\Sigma_i'=(\Sigma^i_c, \Sigma^i_{l} \cup Z_i^{\sim}, \Sigma^i_r)$ and $\Gamma_i'=\Gamma_i
{\times} 2^{S^{\sim}} {\times} \{\first, -\}$, 
$L^0=\{(l^0, \emptyset) \mid l^0 {\in} L^0\}$, and
$F=\{(l^f, \emptyset) \mid l^f {\in} F\}$.
 
%\ram{ $S^{\sim}$ and $Z^{\sim}$ are essentially same. Need both? }  
 
The transitions $\Delta'$ are defined as follows:

\noindent{\emph{Internal Transitions}}.
For every $(l,a,\varphi,l') {\in} \Delta^i_{l}$ we have the set of transitions 
$((l,P),a, \varphi,(l',P)) {\in} {\Delta^i}'_{l}$.
\noindent{\emph{Call Transitions}}.
For every $(l,a,\varphi,l',\gamma) {\in} \Delta^i_c$, we have the following
classes of transitions in $M'$.
\begin{enumerate}
\item
  The first class of transitions corresponds to the guessed pop constraint being
  ${<}\kappa$.
  In the case that, obligation is ${<}\kappa$ is alive in the state, hence there is no  
  need to reset the clock $x_{{<_i}\kappa}$. Otherwise, the obligation ${<}\kappa$ is fresh  
  and hence it is remembered as $\first$ in the $i$th stack, and the clock $x_{{<_i}\kappa}$ is reset.
  \begin{eqnarray*}
    ((l,P), a, \varphi, (l', P), (\gamma,\{{<}\kappa\},-)) {\in} {\Delta^i}'_c && \: \text{if} \:{{<_i}\kappa} {\in} P \\
    ((l,P), a, \varphi, (l'_{a,{{<}\kappa}}, P'), (\gamma,\{{<}\kappa\},\first))
    {\in} {\Delta^i}'_c  && \:\text{if}\: {{<_i}\kappa} {\notin} P \:\text{and}\: P'=P \cup \{{{<_i}\kappa}\}\\
    %((l'_{a,{{<}\kappa}}, P'), {{<_i}\kappa}, x_a=0, (l', P')) {\in} {\Delta^i}'_{int}
    ((l'_{a,{{<}\kappa}}, P'), {{<_i}\kappa}, x_a=0, (l', P')) {\in}
{\Delta^i}'_{l}
  \end{eqnarray*}
\item
The second class of transitions correspond to the guessed pop  
constraint ${>}\kappa$. The clock $x_{{>_i}\kappa}$ is reset, and  
obligation is stored in $i$th stack.
 \begin{eqnarray*} 
   ((l,P), a, \varphi, (l'_{a,{{>}\kappa}}, P), (\gamma,\{{>}\kappa\},-)) {\in}
   {\Delta^i}'_c  &\text{and} &
   %((l'_{a,{{>}\kappa}}, P), {{>_i}\kappa}, x_a{=}0, (l', P)) {\in} {\Delta^i}'_{int} 
   ((l'_{a,{{>}\kappa}}, P), {{>_i}\kappa}, x_a{=}0, (l', P)) {\in}
{\Delta^i}'_{l} 
 \end{eqnarray*}
\item Finally the following transitions consider the case when the 
guessed pop constraint is ${>}\zeta$ and ${<}\kappa$.  
Depending on whether ${<}\kappa$ is alive or not, we have two cases. If alive, then we simply reset the clock $x_{{>_i}\zeta}$  
and remember both the obligations in $i$th stack.   
If ${<}\kappa$ is fresh, then we reset both clocks  
$x_{{>_i}\zeta}$ and $x_{{<_i}\kappa}$ and remember both
obligations in  $i$th stack , and ${<_i}\kappa$ in the state. 
 \begin{eqnarray*}  
    ((l,P), a, \varphi, (l'_{a,{<}\kappa,{>}\zeta}, P'), (\gamma,\{{<}\kappa,
    {>}\zeta\},\first)) {\in} {\Delta^i}'_c && \text{if}\: {{<_i}\kappa} {\notin} P, P'{=}P \cup \{{<_i}\kappa\}\\
    ((l'_{a,{<}\kappa,{>}\zeta},P'), {>_i}\zeta,x_a=0, (l'_{a,{<}\kappa},P')) {\in} {\Delta^i}'_{l}&&\\
%     \ram{ ((l'_{a,{<}\kappa},P'), {>_i}\kappa,x_a=0,
%     (l'_{a,{<}\kappa},P')) {\in} {\Delta^i}'_{l}~\mbox{this reset was
%     missing earlier}}&&\\
     ((l'_{a,{<}\kappa},P'), {>_i}\kappa,x_a=0,
     (l'_{a,{<}\kappa},P')) {\in} {\Delta^i}'_{l}~ &&\\
    ((l,P), a, \varphi, (l'_{a,{>}\zeta}, P), (\gamma,\{{<}\kappa, {>}\zeta\},-)) {\in} {\Delta^i}'_c && \text{if}\: {{<_i}\kappa} {\in} P
 \end{eqnarray*}
\end{enumerate}
\noindent{\emph{Return Transitions}}.
 For every $(l,a,I,\gamma,
\varphi,l') {\in} \Delta^i_r$, transitions in ${\Delta^i}'_r$ are:
\begin{enumerate}
\item $((l,P),a,(\gamma,\{{<}\kappa,{>}\zeta\},-),\varphi \wedge x_{{<_i}\kappa}{<}\kappa \wedge x_{{>_i}\zeta}{>}\zeta,(l',P))$ if $I=(\zeta,\kappa)$.
\item $((l,P),a,(\gamma,\{{<}\kappa,{>}\zeta\},\first),\varphi \wedge x_{{<_i}\kappa}{<}\kappa \wedge x_{{>_i}\zeta}{>}\zeta,(l',P'))$ \\
where $P' = P\backslash\{{<_i}\kappa\}$, if $I=(\zeta,\kappa)$. 
\item $((l,P),a,(\gamma,\{{<}\kappa\},-),\varphi \wedge x_{{<_i}\kappa}{<}\kappa,(l',P))$ if $I=[0,\kappa)$.
\item $((l,P),a,(\gamma,\{{<}\kappa\},\first),\varphi \wedge x_{{<_i}\kappa}{<}\kappa,(l',P'))$ with
$P'{=}P\backslash\{{<_i}\kappa\}$ if $I{=}[0,\kappa)$. 
\item $((l,P),a,(\gamma,\{{>}\zeta\},-),\varphi \wedge x_{{>_i}\zeta}{>}\zeta,(l',P))$ if $I=(\zeta,\infty)$.
\end{enumerate}

For the pop to be successful in $M'$, the guess made at the time of the push must be correct, and 
indeed at the time of the pop, the age must match the constraint.
The control state $(l^f,P)$ is reached in $M'$ on reading a word $w'$ iff 
$M$ accepts a string $w$ and reaches $l^f$. Accepting 
locations of $M'$ are of the form $(l^f,P)$ for $P \subseteq Z^{\sim}$.  
If $w$ is a matched word then $P$ is empty, and there must be any
obligations which are pending at the end.

Let $w=(a_1,t_1) \dots (a_i,t_i) \dots (a_n,t_n) {\in} L(M)$.
If $a_i \in \Sigma^i_c$, 
 we have in $L(M')$, a string $T_i$ between $(a_i,t_i)$ and $(a_{i+1},t_{i+1})$,  with 
 $|T_i| \leq 2$, and $T_i$ is a timed word of the form 
 $(b_{1i},t_i)(b_{2i},t_i)$ or  $(b_{1i},t_i)$.
The time stamp $t_i$ remains unchanged, and either $b_{1i}$ is $<_i\kappa$ or $\leq_i \kappa$  
or $b_{1i}$ is $>_i\zeta$, or 
$b_{1i}$ is  $>_i\zeta$ and $b_{2i}$ is one of $<_i \kappa$ or $\leq_i \kappa$  
 for some $\kappa,\zeta \leq k$. This follows from the three kinds of call transitions in $M'$.

%This untiming construction is given in the detailed proof which can be
%found in Appendix~\ref{app:untime}}. 
 
In the construction above, 
it can shown by inducting on the length of words accepted that $h(L(M'))=L(M)$. 
Thus, $L(M') \neq \emptyset$ iff $L(M) \neq \emptyset$. 
If $M$ is a \ram{$k$-scope-\dtecmvpa{}, then $M'$ is a
$k$-scope-\ECMVPA{}}.
\ram{Since $M'$ is a $k$-scope-\ECMVPA{},  its emptiness check is
decidable using Theorem~\ref{lab:tm:empt-kecmvpa}, which uses 
the standard region
construction of event clock automata \cite{AFH99} to obtain a
$k$-scope-\MVPA{},}
which has a decidable emptiness \cite{latin10}.

\begin{theorem}
\label{lab:tm:empt}
%The emptiness checking problem for $k$-scope \dtecmvpa{} is decidable. 
The emptiness checking for $k$-scope \dtecmvpa{} is decidable. 
\end{theorem}

In \cite{BDKPT16} we have shown that $k$-round \dtecmvpa{} are determinizable. 
Using an untiming construction to get $k$-round \ecmvpa, determinize
it and again converting this to get deterministic $k$-round \dtecmvpa{}. 
 
     For $k$-scope \dtecmvpa{} using the stack untiming construction we get 
a $k$-scope \ecmvpa. This is determinized to get deterministic $k$-scope \ecmvpa.  
We convert this back to get deterministic $k$-scope \dtmvpa{}.  
The morphisms used for this conversion are same as in \cite{BDKPT16}. 

%We give a proof sketch of the following theorem whose details are found in
%Appendix~\ref{lab:app:det}.

\begin{theorem}
\label{lab:tm:kdtdet}
The $k$-scope \dtecmvpa{} are determinizable. 
\end{theorem}
\begin{proof} 
Consider a $k$-scope \dtecmvpa{} $M=(L, \Sigma, \Gamma, L^0, F, \Delta)$ and the corresponding  
$k$-scope \ecmvpa{} $M'=(L', \Sigma', \Gamma', L'^0, F', \Delta')$
 as constructed in section \ref{app:untime}.
From Theorem~\ref{lab:tm:det-kecmvpa} we know that $M'$ is determinizable. 
Let $Det(M')$ be the determinized 
automaton such that $L(Det(M'))=L(M')$. That is,
$L(M)=h(L(Det(M')))$.  By construction of $M'$, 
we know that the new symbols introduced in $\Sigma'$ are $Z^{\sim}$
($\Sigma'_i=\Sigma_i \cup  Z_i^{\sim}$ for each $i$th stack )
and 
(i) no time elapse happens on reading symbols from $Z_i^{\sim}$, and 
(ii) no stack operations happen on reading  symbols of $Z_i^{\sim}$.  
Consider any transition in $Det(M')$ involving 
the new symbols. Since $Det(M')$ is deterministic, 
let $(s_1, \alpha, \varphi, s_2)$ be 
the unique transition on $\alpha {\in} Z_i^{\sim}$. In the following, we 
eliminate these transitions on  $Z_i^{\sim}$ preserving the language accepted by $M$ and 
the determinism of $det(M')$. In doing so, we will 
%construct a $k$-\dtMVPA{} $M''$ which is deterministic, and 
construct a \ram{$k$-scope \dtecmvpa{}} $M''$ which is deterministic, and 
which preserves the language of $M$. 
We now analyze various types for $\alpha {\in}   Z_i^{\sim}$. 

\begin{enumerate}
\item 
Assume that $\alpha$ is of the form ${>_i}\zeta$. 
Let $(s_1, \alpha, \varphi, s_2)$ be 
the unique transition on $\alpha {\in} Z_i^{\sim}$. By construction of $M'$ (and hence $det(M')$),  we 
know that $\varphi$ is $x_a=0$ for some 
$a {\in} \Sigma^i$. We also know that in $Det(M')$, there 
is a unique transition  
$(s_0, a, \psi, s_1, (\gamma, \alpha, -))$ preceding $(s_1, \alpha, \varphi, s_2)$.
 Since $(s_1, \alpha, \varphi, s_2)$ is a no time elapse transition, and does not touch any stack, 
 we can combine the two transitions from $s_0$ to $s_1$ and $s_1$ to $s_2$ to
 obtain  the call transition % $(s_0, a, \psi, s_2, (\gamma, \alpha, -))$ for $i$th stack .  
$(s_0, a, \psi, s_2, \gamma)$ for $i$th stack .  
 This eliminates transition on ${>_i}\zeta$. 

\item 
  Assume that $\alpha$ is of the form ${<_i}\kappa$.
Let $(s_1, \alpha, \varphi, s_2)$ be 
the unique transition on $\alpha {\in} Z_i^{\sim}$.
We also know that $\varphi$ is $x_a=0$ for some 
$a {\in} \Sigma^i$. From $M'$, we also know that in $Det(M')$, there 
is a unique transition of one of the following forms preceding $(s_1, \alpha, \varphi, s_2)$:
\begin{itemize}
\item[(a)] $(s_0, a, \psi, s_1, (\gamma, \alpha, -))$, 
(b) $(s_0, a, \psi, s_1, (\gamma, \alpha,\first))$, or 
\item[(c)] $(s_0, {>_i}\zeta, \varphi, s_1)$ where it is preceded by 
 $(s'_0,a, \psi, s_0,  (\gamma, \{\alpha, {>}\zeta\},X))$ for $X {\in} \{\first,-\}$. 
\end{itemize}
As $(s_1, \alpha, \varphi, s_2)$ is a no time elapse transition, 
and does not touch the stack, we can combine two transitions from  
$s_0$ to $s_1$ (cases $(a)$, $(b)$) and
$s_1$ to $s_2$ to obtain  the call transition
$(s_0, a, \psi, s_2, (\gamma, \alpha, -))$ or
$(s_0, a, \psi, s_2, (\gamma, \alpha, \first))$.   
This eliminates the transition on ${<_i}\kappa$. 

In case of transition $(c)$, we first eliminate 
the local  transition on  ${>_i}\zeta$ obtaining 
%$(s'_0,a, \psi, s_1,  (\gamma, \{\alpha, {>}\zeta\}, X))$. 
$(s'_0,a, \psi, s_1,  \gamma)$. 
This can then be combined with $(s_1, \alpha, \varphi, s_2)$ to obtain 
 the  call transitions 
%$(s'_0,a, \psi, s_2,  (\gamma, \{\alpha, {>}\zeta\},X))$. 
$(s'_0,a, \psi, s_2,  \gamma)$. 
We have eliminated local transitions on ${<_i}\kappa$. 
\end{enumerate}
  
Merging transitions as done here does not affect transitions on any $\Sigma^i$ as
they simply eliminate the newly added transitions on
$\Sigma_i'\setminus \Sigma_i$. 
Recall that checking constraints on recorders $x_{<_i \kappa}$ and $x_{>_i \zeta}$ were required during return  transitions. 
We now modify the pop operations in $Det(M')$ as follows:
Return transitions have the following forms, and in all of these, $\varphi$ is a constraint 
checked on the clocks of $C_{\Sigma^i}$ in $M$ during return: 
transitions $(s, a, (\gamma, \{{<}\kappa\}, X), \varphi \wedge
  x_{{<_i}\kappa} {<}\kappa, s')$ for $X {\in} \{-, \first\}$ are modified to 
  %$(s, a, [0, \kappa), (\gamma, \{{<}\kappa\}, X), \varphi, s')$;
  $(s, a, [0, \kappa), \gamma, \varphi, s')$;
transitions $(s, a, (\gamma, \{{<}\kappa, {>}\zeta\}, X), \varphi
  \wedge x_{{>_i}\zeta} {{>}} \zeta \wedge x_{{<_i}\kappa} {{<}} \kappa, s')$ for
  $X {\in} \{-, \first\}$ are modified to   
%$(s, a, (\zeta, \kappa), (\gamma, \{{<}\kappa, {>}\zeta\}, X), \varphi, s')$; and 
$(s, a, (\zeta, \kappa), \gamma, \varphi, s')$; and 
transition $(s, a, (\gamma, \{{>}\zeta\}, -), \varphi \wedge
  x_{{>_i}\zeta} {>} \zeta, s')$ are modified to the transitions 
  %$(s, a, (\zeta, \infty), (\gamma, \{{>}\zeta\}, -), \varphi, s')$.
  $(s, a, (\zeta, \infty), \gamma, \varphi, s')$.
Now it is straightforward to verify that the $k$-scope $\dtecmvpa{}$ $M''$
obtained from the $k$-scope \ecmvpa{} $det(M')$ is
deterministic. Also, since we have only eliminated symbols of $Z^{\sim}$, we have
  $L(M'')=L(M)$ and $h(L(M''))=L(det(M'))$.  
This completes the proof of determinizability of $k$-scope \dtecmvpa{}.
 
\end{proof} 
 
It is easy to show that $k$-scoped \ecmvpa s and $k$-scoped 
\dtecmvpa s are closed under union and intersection; using  
Theorem \ref{lab:tm:det-kecmvpa} and Theorem \ref{lab:tm:kdtdet} 
we get closure under complementation. 
\begin{theorem}      
    \label{lab:tm:boolmain}
%The classes of $k$-scoped \ecmvpl s and $k$-scoped \dtecmvpl s  
The classes of $k$-scoped \ecmvpl s and $k$-scoped \ram{\dtecmvpl{}s}  
    are closed under Boolean operations. 
\end{theorem}

\section{Logical Characterization of $k$-\dtecmvpa{}}
\label{sec:mso}
%\input{mso.tex}

%We consider a timed word  {$w = (a_1, t_1), \ldots, (a_m, t_m)$} 
Let $w = (a_1, t_1), \ldots, (a_m, t_m)$ 
be a timed word 
%{$w = (a_0, t_0), (a_1, t_1), \ldots, (a_m, t_m)$}  
over alphabet 
%$\Sigma=\langle \Sigma^i_c, \Sigma^i_{int}, \Sigma^i_r \rangle_{i=1}^n$  as a \emph{word structure} over the universe $U=\{1,2,\dots,|w|\}$ 
$\Sigma=\langle \Sigma^i_c, \Sigma^i_{l}, \Sigma^i_r \rangle_{i=1}^n$  as a \emph{word structure} over the universe $U=\{1,2,\dots,|w|\}$ 
of positions in $w$. 
We borrow definitions of predicates $Q_a(i), \ERPred{a}(i),
\EPPred{a}(i)$ from ~\cite{BDKRT15}.
% The predicates in the word structure are $Q_a(i)$ for $a \in \Sigma$ which evaluates to true at position $i$ iff \replaced[id=R3]{$w[i]=(a, \cdot)$}{$w[i]=a$}\Review{R3.B26a}, where $w[i]$ denotes the $i$th position of $w$. 
Following~\cite{lics07}, we use the matching binary relation $\mu_j(i,k)$ 
which evaluates to true iff the $i$th position is a call and the $k$th position
is its matching return corresponding to the $j$th stack. 
% We also introduce three predicates $\ERPred{a}$, $\EPPred{a}$, and $\theta_j$ capturing the following relations.
% For an interval $I$, the predicate $\ERPred{a}(i) \in I$ evaluates to true on the word structure iff $\nu^{w}_i(x_a) \in  I$ for recorder clock $x_a$. 
% For an interval $I$, the predicate $\EPPred{a}(i) \in I$ evaluates to true on the word structure iff $\nu^{w}_i(y_a) \in  I$ for predictor clock $y_a$.
We introduce the predicate $\theta_j(i) \in  I$ which evaluates to true on the
word structure  iff {$w[i] = (a, t_i)$ with $a \in \Sigma^j_r$} and {$w[i] \in \Sigma^j_r$}, and there is some $k<i$ such that 
$\mu_j(k,i)$ evaluates to true  and $t_i-t_k \in  I$.  
The predicate $\theta_j(i)$ measures time elapsed between position $k$
where a call was made on the stack $j$, and position $i$, its matching return. 
This time elapse is the age of the symbol pushed onto the stack during the call
%at position $k$. Since position $i$ is the matching return, this symbol is popped at position $i$;
at position $k$. Since position $i$ is the matching return, this
symbol is popped at $i$,
if the age lies in the interval $I$, the predicate evaluates to true.  
We define MSO($\Sigma$), the MSO logic over $\Sigma$, as:
\[
\varphi{:=}Q_a(x)~|~x {\in} X~|~\mu_j(x,y)~|~\ERPred{a}(x){\in} I~|~\EPPred{a}(x){\in}
I~|~\theta_j(x){\in} I~|\neg\varphi~|~\varphi {\vee}
\varphi|~\exists \, x.\varphi~|~\exists \, X.\varphi
\]
where $a {\in} \Sigma$, $x_a {\in} C_{\Sigma}$, $x$ is a first order variable and
$X$ is a second order variable. 

The models of a formula $\phi \in \mbox{MSO}(\Sigma)$ are timed words $w$ over
$\Sigma$. 
The semantics %of this logic  
is standard where first order variables are
interpreted over positions  of $w$ and second order variables over subsets of
positions.
We define the language $L(\varphi)$ of an MSO sentence $\varphi$ as the set of
all words satisfying $\varphi$.

%Words in $Round(\Sigma,k)$, for some $k$ rounds, can be captured by an MSO formula $Bd_k(\psi)$.
%For instance if $k=1$, and $n$ stacks, the formula
%$\exists x_1. (Q_{{a}^1}(x_1) \wedge \forall y_1(y_1 \leq x_1 \rightarrow Q_{{a}^1}(y_1)) 
%\wedge \exists x_2.  (x_1 < x_2 \wedge Q_{{a}^2}(x_2) \wedge \forall
%y_2(x_1<y_2<x_2 \rightarrow Q_{{a}^2}(y_2)) \wedge  \ldots \wedge \exists x_{n}(x_{n-1} < x_{n} \wedge
%Q_{{a}^n}(x_n) \wedge last(x_n) \wedge \forall
%y_n(x_{n-1}<y_n<x_n \rightarrow Q_{{a}^n}(y_n))) ))$, 
%where ${a}^i \in \Sigma^i$ and $last(x)$ denotes $x$  is the last position,
%captures a round. This can be extended to capture $k$-round words. 
%Conjuncting the formula obtained from a \dtMVPA{} $M$ with  $Bd_k(\psi)$ accepts only those words 
%which lie in $L(M) \cap Round(\Sigma,k)$. Likewise, if one considers any MSO formula 
%$\zeta=\varphi \wedge  Bd_k(\psi)$, it can be shown that the \dtMVPA{} 
%$M$ constructed for $\zeta$ will be a $k$-\dtMVPA{}.  
% 
% 
%
 
%\ram{ 
Words in $Scope(\Sigma,k)$, for some $k$, can be captured by an MSO
%formula $Scope_k(\psi)$, given as follows.\\
formula $Scope_k(\psi)= \andover_{1 \leq j \leq n} Scope_k(\psi)^j$, where $n$ is
number of stacks, where \\
\begin{center}
$Scope_k(\psi)^j= \forall y Q_a(y) \wedge a \in \Sigma^r_j 
    \Rightarrow( \exists x\mu_j(x,y) \wedge $
$( \psi_{kcnxt}^j \wedge \psi_{matcnxt}^j \wedge
\psi_{noextracnxt} 
    ))$
\end{center}
where
$\psi_{kcnxt}^j = \exists x_1,\ldots,x_k (x_1 \leq \ldots \leq
x_k \leq y \andover_{1 \leq q \leq k} (Q_a(x_q) \wedge a \in \Sigma_j
\wedge (Q_b(x_q-1) \Rightarrow b \notin \Sigma_j))$, 
 
\noindent 
and  
$\psi_{matcnxt}^j=
\orover_{1\leq q \leq k} \forall x_i (x_q \leq x_i \leq x (Q_c(x_i)
\Rightarrow c \in \Sigma_j)) 
$, and \\
$\psi_{noextracnxt}=
\exists x_l (x_1 \leq x_l \leq y) (Q_a(l) \wedge a \in \Sigma_j \wedge
Q_b(x_{l}-1) \wedge b \in \Sigma_j) \Rightarrow x_l \in
\{x_1,\ldots,x_k\}.$ 
 
Formulas $\psi_{noextracnxt}$ and $\psi_{kcnxt}$ say that there are at most $k$
contexts of $j$-th stack, while formula $\psi_{matcnxt}$ says where 
matching call position $x$ of return position $y$ is found. 
%} 
Conjuncting the formula obtained from a \dtecmvpa{} $M$ with  $Scope(\psi)$ accepts only those words 
which lie in $L(M) \cap Scope(\Sigma,k)$. Likewise, if one considers any MSO formula 
$\zeta=\varphi \wedge  Scope(\psi)$, it can be shown that the \dtecmvpa{} 
$M$ constructed for $\zeta$ will be a $k$-\dtecmvpa{}.

Hence we have the following MSO characterization. 

\begin{theorem}
\label{th:mso}
 A language $L$ over $\Sigma$ is accepted by an \name{k-scope \dtecmvpa{}} iff there is a
 MSO sentence $\varphi$ over $\Sigma$ such that $L(\varphi) \cap Scope(\Sigma,k)=L$. 
\end{theorem}

The two directions, \dtecmvpa{} 
%to MSO, as well as MSO to \dtMVPA{} can be handled using standard techniques, and can be found in Appendix \ref{app:mso}.
to MSO, as well as MSO to \dtecmvpa{} can be handled using standard
techniques, and can be found in Appendix~\ref{app:mso}.
 
% to MSO, as well as MSO to \dtMVPA{} can be handled using standard
% techniques, and is given below. 

\section{Conclusion}
In this work we have seen timed context languages characterized by $k$-scope  
\ecmvpa{} and dense-timed $k$-scope \dtecmvpa{}, along with logical characterizations  
for both classes (for $k$-scope \ecmvpa{}, equivalent MSO is obtained  
by dropping predicate $\theta_j(x){\in} I$--which checks if the aging time of  
pushed symbol $x$ of stack $j$~falls in interval $I$--from MSO of 
$k$-scope \dtecmvpa{}). 

Here, while reading an input symbol of a stack, clock constraints
involve the clocks of symbols associated with the same stack. It would
be interesting to see if our results hold
without this restriction.   
%Also, we would like to simplify
%determinization procedure  for $k$-scope \dtecmvpa{}.  
Another direction, would be to use alternate methods like Ramsey
theorem based \cite{Abdulla11,ramsey-vpa} or Anti-chain
based \cite{FogartyV10,BruyereDG13} methods avoiding  
complementation and hence determinization to do language inclusion check.

\bibliographystyle{splncs04}
\bibliography{ms}

\newpage
\appendix
\centerline{\Large \bf Appendix}

\section{Case study: Modeling A single channel packet switching communication network}
A single channel packet switching communication network uses a common medium to send fixed sized packets from a transmitter to one or many receivers.
Such medium may be a bus or a cable or any other transmission medium.
The communication network consists of multiple stations sharing a single error free communication channel supported by the medium.
All stations transmit fixed sized packets of data over the channel and stations receive them.
The time is slotted into frames and the transmission of new packet can start only at the beginning of the frame.
It is possible that one than one stations may transmit packets simultaneously, and in that case packets are said to collide.
Such collision results in the transmission error and whether a collision has occured or not is decided at the end of frame.
If the collision has happended, it is visible to all stations and failure of communication is concluded.
If no collision occurs, packet is assumed to be transmitted successfully.
All stations employ same protocol to avoid collision and retransmission.
Such choice of such protocol determines various characteristics of the network like stability and maximum throughput.
%Here we model the Capetanakis-Tsybakov-Mikhailov (CTM) protocol~\cite{TsyVve80,1057014} using a \dtecmvpa{} shown in Figure~\ref{fig:ctm}. It shows state transitions for station $i$. The actual automaton for CTM protocol is constructed by taking cross product of automata for all stations.
Here we model the Capetanakis-Tsybakov-Mikhailov (CTM) protocol~\cite{TsyVve80} using a \dtecmvpa{} shown in Figure~\ref{fig:ctm}. It shows state transitions for station $i$. The actual automaton for CTM protocol is constructed by taking cross product of automata for all stations.

\begin{figure}[h]
 \includegraphics{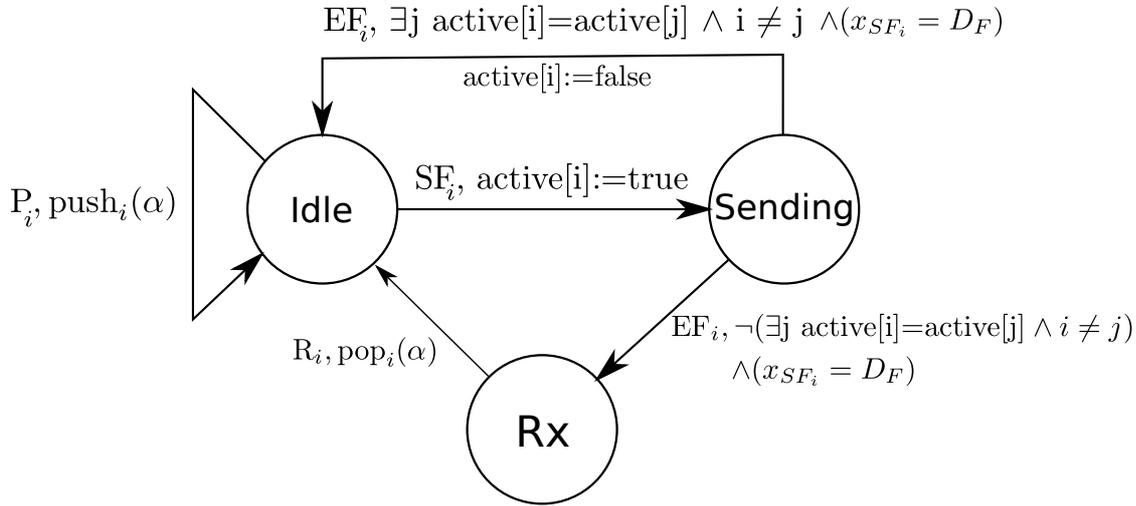}
 \caption{Model of CTM protocol for station $i$}
 \label{fig:ctm}
\end{figure}

% We denote start of each message sent by station $i$ using the symbol $\text{SF}_i$ and end using the symbol $\text{EF}_i$.
For each station $i$, the symbol $\text{SF}_i$, $\text{EF}_i$, $\text{P}_i$ and $\text{R}_i$ denote the start of frame, end of frame, arrival of new packet and receipt of packet respectively. We use a global shared variable array \textsf{active} to record stations which are actively sending packets.
\textsf{active[i]} is true iff station $i$ is sending packet.
We detect the occurance of collision if more that one entry in \textsf{active} is true.
The constant $D_F$ denote the durations of frame (i.e. time difference between \text{SF} and \text{EF}).
There is an idle period between two successive frames (i.e. time difference between \text{EF} for first frame and \text{SF} of the next frame). We denote it using constant $D_I$.
The event clocks have been used to measure such timing requirements.
CTM algorithm expects that start and end of frame is synchronized for all stations.
Such requirement can be taken care by introducing extra symbols during initialization phase (not shown in Figure~\ref{fig:ctm}).
We introduce additional return input symbol $P'_0$.
At global time zero input symbol $P_0$ occurs which pushes $\alpha_0$ on the stack at station 0.
We then require that input symbols $\text{SF}_0 \dots \text{SF}_n$ occur.
We then mandate $P'_0$ to occur when the age of popped $\alpha_0$ is zero.
This ensures that the all $\text{SF}_i$ are synchronized when global time is zero.

\section{Details of Theorem \ref{th:mso}}
\label{app:mso}
%Here, we give the details of the translations from \dtMVPA{} to MSO and
Here, we give the details of the translations from
\ram{\dtecmvpa{}} to MSO and conversely.  
A technical point is regarding the projection operation: in general, it is
known that event clock automata (hence \dtecmvpa{}) are not closed under
%known that event clock automata (hence \dtMVPA{}) are not closed under
projections.  
However, we need to handle projections while quantifying out variables 
%in the MSO to \dtMVPA{} construction.
in the MSO to \ram{\dtecmvpa{}} construction.
We do this by working on Quasi \ram{\dtecmvpa{}} where the underlying alphabet $\Sigma$
is partitioned into finitely many buckets  
 $P_1, \dots, P_k$ via   a ranking function $\rho: \Sigma \rightarrow \mathbb{N}$.
All symbols in a $P_j$ are then ``equivalent" : we assign one event recorder and  
one event predictor per $P_i$. This helps in arguing the correctness of the
constructed \ram{\dtecmvpa{}} from an MSO formula while projecting out variables. 
In Section \ref{app:quasi}, we show the equi-expressiveness of  
%quasi \dtMVPA{} and \dtMVPA{} which allows us to complete the logical characterization.
\ram{quasi \dtecmvpa{} and \dtecmvpa{}} which allows us to complete the logical characterization.
\begin{itemize}
  \item {\bf Logic to automata.} 
    We first show that the language accepted by an  MSO formula $\varphi$ over
    %$\Sigma=\langle\Sigma^i_c, \Sigma^i_{int}, \Sigma^i_r\rangle_{i=1}^n$,
    $\Sigma=\langle\Sigma^i_c, \Sigma^i_{l}, \Sigma^i_r\rangle_{i=1}^n$,
    %$L(\varphi)$ is accepted by a \dtMVPA{}.  
    $L(\varphi)$ is accepted by a \dtecmvpa{}.  
    Let $Z=(x_1, \dots, x_m, X_1, \dots, X_n)$ be the free variables in
    $\varphi$. 
    As usual, we work on the extended alphabet
    %$\Sigma'=\langle{\Sigma^i}'_c, {\Sigma^i}'_{int}, {\Sigma^i}'_r\rangle_{i=0}^n$ 
    $\Sigma'=\langle{\Sigma^i}'_c, {\Sigma^i}'_{l}, {\Sigma^i}'_r\rangle_{i=0}^n$ 
    where
    \[
      {\Sigma^i}'_s{=}\Sigma^i_s {\times} (Val:Z \rightarrow \{0,1\}^{m+n}),
    \]
    for $s   \in \{c,l,r\}$.  
    %for $s   \in \{c,int,r\}$.  
    A word $w'$ over $\Sigma'$ encodes a word over $\Sigma$ along with the valuation of all first order and second order variables.  
    Thus ${\Sigma^i}'$ consists of all symbols $(a,v)$ where $a \in \Sigma^i$ is such
    that $v(x)=1$ means  that $x$ is assigned the position $i$ of $a$ in the word $w$, while $v(x)=0$ means that $x$ is not assigned 
    the position of $a$ in $w$.  
    Similarly, $v(X)=1$ means that the position $i$ of $a$ in $w$ belongs to the set $X$. 
    Next, we use quasi-event clocks for $\Sigma'$ by assigning suitable ranking function.
    %Quasi \dtMVPA{} are equiexpressive to \dtMVPA{} as explained in Section \ref{app:quasi}.
    Quasi \dtecmvpa{} are equiexpressive to \dtecmvpa{} as explained in Section \ref{app:quasi}.
    We partition each ${\Sigma^i}'$ such that for a fixed $a \in \Sigma^i$, all symbols of the
form $(a,d_1, \dots, d_{m+n})$ and $d_i \in \{0,1\}$ lie in the same partition  
($a$ determines their partition).
Let $\rho':\Sigma' \rightarrow \mathbb{N}$ be the ranking function of $\Sigma'$
wrt above partitioning scheme. 

Let $L(\psi)$ be the set of all words $w'$ over $\Sigma'$ such that the
underlying word $w$ over $\Sigma$ satisfies formula $\psi$ along with the
valuation $Val$.  
Structurally inducting over $\psi$, we show that $L(\psi)$ is accepted by a
\ram{\dtecmvpa{}.}
%\dtMVPA{}.
The cases   $Q_a(x), \mu_j(x,y)$ are exactly as in \cite{lics07}. We only
discuss the predicate $\theta_j$  here.  
Consider the atomic formula $\theta_j(x) \in I$. To handle 
this, we build a {\dtecmvpa{}} that keeps pushing 
%this, we build a \dtMVPA{} that keeps pushing 
symbols $(a,v)$ onto the stack $j$ whenever $a \in \Sigma^j_c$, initializing 
the age to 0 on push. It keeps popping 
the stack on reading return symbols $(a',v')$, $a' \in \Sigma^j_r$, and checks whether 
$v'(x)=1$ and $age(a',v') \in I$. It accepts on finding 
such a pop. The check  $v'(x)=1$ ensures that this is the matching return 
of the call made at position $x$. The check  $age(a',v') \in I$
confirms that the age of this symbol pushed at position $x$ is indeed in the interval $I$. 
Negations, conjunctions, and disjunctions follow from the closure properties of 
%\dtMVPA{}. 
\ram{\dtecmvpa{}. }

Existential quantifications correspond to a projection by excluding the chosen variable from the valuation and renaming the alphabet $\Sigma'$. 
Let $M$ be a \ram{\dtecmvpa{}} constructed for $\varphi(x_1, \dots,x_n,X_1, \dots, X_m)$ 
over $\Sigma'$.
Consider $\exists x_i. \varphi(x_1, \dots,x_n,X_1, \dots, X_m)$ for some first
order variable $x_i$.  
Let $Z_i=(x_1, \dots, x_{i-1},x_{i+1}, \dots,x_n, X_1, \dots, X_m)$ by removing $x_i$ from $Z$. 
We simply work on the alphabet $\Sigma'{\downarrow} i=\Sigma \times (Val:Z_i \rightarrow \{0,1\}^{m+n-1})$.
Note that $\Sigma'{\downarrow} i$ is partitioned exactly in the same way as $\Sigma'$. For a fixed $a \in \Sigma$, 
all symbols $(a,d_1, \dots, d_{m+n-1})$ for $d_i \in \{0,1\}$ lie in the same partition.
Thus, $\Sigma'$ and $\Sigma'{\downarrow} i$ have exactly the same number of partitions, namely
$|\Sigma|$.
  Thus, an event clock $x_a=x_{(a,d_1, \dots, d_{m+n})}$ 
 used in $M$ can be used the same way while constructing the automaton 
 for $\exists x_i. \varphi(x_1, \dots,x_n,X_1, \dots, X_m)$. The case of $\exists X_i. \varphi(x_1, \dots,x_n,X_1, \dots, X_m)$
is similar.
Hence we obtain in all cases, a \ram{\dtecmvpa{}} that accepts $L(\psi)$ when 
$\psi$ is an MSO  sentence.

\item {\bf Automata to logic.}
%Consider a \dtMVPA{} $M=(L, \Sigma, \Gamma, L^0, F, \Delta)$. 
Consider a \dtecmvpa{} $M=(L, \Sigma, \Gamma, L^0, F, \Delta)$. 
For each stack $i$, let $C^i_{\gamma}$ denote a second order variable
which collects all positions where $\gamma$ is pushed in stack $i$. Similarly, 
let $R^i_{\gamma}$ be a second order variable which collects all positions where 
$\gamma$ is popped from stack $i$. Let $X_{l_i}$ be a second order variable which collects all positions where the location is $l_i$ in a run. Let $\mathcal{C}$, $\mathcal{R}$ and
$\mathcal{L}$ respectively be the set of these variables.  

%The MSO formula encoding runs of the \dtMVPA{}  is:
The MSO formula encoding runs of the \ram{\dtecmvpa{}}  is:
$\exists \mathcal{L} ~\exists \mathcal{C} ~\exists \mathcal{R}
~\varphi(\mathcal{L}, \mathcal{C}, \mathcal{R})$.
We assert that the starting position must belong to $X_{l}$ for some $l \in
L^0$. Successive positions must be connected by an appropriate transition. To complete
the reduction we list these constraints.
\begin{itemize}
\item
  For call transitions $(\ell_i, a, \psi, \ell_j, \gamma) \in \Delta^h_c$, for
  positions $x,y$, assert 
\hspace{-8em}  \[
  X_{\ell_i}(x) \wedge X_{\ell_j}(y) \wedge Q_a(x) \wedge C^h_\gamma(x) 
    \wedge  
  \]
\[ 
\hspace{8em} 
\bigwedge_{b \in \Sigma^h} \Big ( \big ( \bigwedge_{(x_b \in I) \in \psi}
    \ERPred{b}(x) \in I \big ) \wedge \big ( \bigwedge_{(y_b \in I) \in \psi}
    \EPPred{b}(x) \in I \big ) \Big ).
    \]
\item
  For return transitions $(\ell_i, a, I, \gamma, \psi, \ell_j) \in \Delta^h_r$ for
  positions $x$ and $y$ we assert that
  \[
  X_{\ell_i}(x) \wedge X_{\ell_j}(y) \wedge Q_a(x) \wedge R^h_\gamma(x) \wedge \theta^h(x) {\in} I 
  \wedge \]
  \[
%\quad \quad \quad \quad \quad \quad \quad \quad \bigwedge_{b \in \Sigma^h} \Big ( \big ( \bigwedge_{(x_b \in I) \in \psi}
\hspace{8em} \bigwedge_{b \in \Sigma^h} \Big ( \big ( \bigwedge_{(x_b \in I) \in \psi}
  \ERPred{b}(x) {\in} I \big ) \,{\wedge} \big ( \bigwedge_{(y_b \in I) {\in} \psi}
  \EPPred{b}(x) {\in} I \big ) \Big ).
  \]
  \item
  %Finally, for internal transitions $(\ell_i, a, \psi, \ell_j) \in \Delta^h_{int}$ for
  Finally, for internal transitions $(\ell_i, a, \psi, \ell_j) \in
\Delta^h_{l}$ for
  positions $x$ and $y$ we assert
  \[
  X_{\ell_i}(x) \wedge X_{\ell_j}(y) \wedge Q_a(x) \wedge
    \bigwedge_{b \in \Sigma^h} \Big ( \big ( \bigwedge_{(x_b \in I)
\in \psi} \hspace*{-.4cm}
    \ERPred{b}(x) \in I \big ) \wedge \big ( \bigwedge_{(y_b \in I) \in \psi} \hspace*{-.4cm}
    \EPPred{b}(x) \in I \big ) \Big ).
    \]
  \end{itemize}
  We also assert that the last position of the word belongs to some $X_l$ such that 
there is a transition (call, return, local) from $l$ to an accepting location. The encoding 
of all three  kinds of transitions are given as above. 
Additionally, we assert that corresponding call and return positions should
match, i.e. 
\[
\forall x \forall y \, \mu_j(x, y) \Rightarrow \bigvee_{\gamma \in \Gamma^j
  \setminus \bott_j} C^j_\gamma(x) \wedge R^j_\gamma(y).
\]
\end{itemize}

%\subsection{Remaining part: Quasi \dtMVPA{}}
\subsection{Remaining part: Quasi \ram{\dtecmvpa{}}}
\label{app:quasi}
%A quasi \name{k-\dtMVPA{}} is a weaker form of \name{k-\dtMVPA{}} where more than one input symbols share the same event clock.
A quasi \name{k-\dtecmvpa{}} is a weaker form of \name{k-\dtecmvpa{}} where more than one input symbols share the same event clock.
Let the finite input alphabet $\Sigma$ be partitioned into finitely many classes
via a ranking function $\rho:  \Sigma \to  \Nat$ giving rise to finitely many
partitions $P_1, \dots, P_k$ of $\Sigma$ where 
$P_i = \set{a \in \Sigma \mid \rho(a)=i}$.
The event recorder $x_{P_i}$ records the time elapsed since the last occurrence
of some action in $P_i$, while the event predictor $y_{P_i}$ predicts the time
required for any action of $P_i$ to occur. 
Notice that since clock resets are ``visible'' in input timed word, the
clock valuations after reading a prefix of the word are also determined by the
timed word.  

%\begin{definition}[Quasi \name{k-\dtMVPA{}}]
\begin{definition}[Quasi \name{k-\dtecmvpa{}}]
%A quasi $k$-\dtMVPA{} over $\Sigma
A quasi \ram{$k$-\dtecmvpa{}} over alphabet $\Sigma
%{=} \set{\Sigma^i_c, \Sigma^i_r, \Sigma^i_{int}}_{i=1}^n$ is a tuple $M {=} 
{=} \set{\Sigma^i_c, \Sigma^i_r, \Sigma^i_{l}}_{i=1}^n$ is a tuple $M {=} 
  (L, \Sigma, \rho, \Gamma, L^0, F, \Delta)$ 
where $L$ is a finite set of locations
including a set $L^0 \subseteq L$ of initial locations, 
$\rho$ is the ranking function,
$\Gamma$ is the stack alphabet
and $F {\subseteq} L$ is a set of final locations.\end{definition}

\begin{lemma}
\label{lab:lm:qdtmvpa}
 %Quasi \name{k-\dtMVPA{}} and \name{k-\dtMVPA{}} are effectively equivalent.
 Quasi \name{k-\dtecmvpa{}} and \name{k-\dtecmvpa{}} are effectively equivalent.
\end{lemma}
%The proof of Lemma~\ref{lab:lm:qdtmvpa} is obtained by using the construction 
%proposed in the proof of Lemma~\ref{lab:lm:qdtvpa}. 

The proof of Lemma~\ref{lab:lm:qdtmvpa} is obtained by using the construction 
proposed in the proof of Lemma~\ref{lab:lm:qdtvpa1} for single stack
machines. 
 
\begin{lemma}[\cite{BDKRT15}]        
The class of $q$-\dtvpa{} is equally expressive as the class of $\dtvpa$. 
\label{lab:lm:qdtvpa1} 
\end{lemma}

\begin{paragraph}{Valid homomorphisms for quasi \name{k-\dtMVPA{}}}
%Let $\Sigma = \set{\Sigma^i_c, \Sigma^i_r, \Sigma^i_{int}}_{i=1}^n$ and $\Pi = \set{\Pi^i_c, \Pi^i_r, \Pi^i_{int}}_{i=1}^n$
Let $\Sigma = \set{\Sigma^i_c, \Sigma^i_r, \Sigma^i_{l}}_{i=1}^n$ and
$\Pi = \set{\Pi^i_c, \Pi^i_r, \Pi^i_{l}}_{i=1}^n$
be the alphabets of \name{k-\dtMVPA{}}s $M_1 = (L_1, \Sigma, \rho_1, \Gamma_1, L^0_1, F_1, \Delta_1)$ and $M_2 = (L_2, \Pi, \rho_2, \Gamma_2, L^0_2, F_2, \Delta_2)$ respectively.
A homomorphism $h:\Sigma \mapsto \Pi$ is said to \emph{valid} iff following conditions are satisfied
 \begin{itemize}  
  %\item $h$ preserves stack mapping $i.e.$ $a \in \Sigma^i_c$ iff $h(a) \in \Pi^i_c$, $b \in \Sigma^i_r$ iff $h(b) \in \Pi^i_r$ and $c \in \Sigma^i_{int}$ iff $h(c) \in \Pi^i_{int}$
  \item $h$ preserves stack mapping $i.e.$ $a \in \Sigma^i_c$ iff
$h(a) \in \Pi^i_c$, $b \in \Sigma^i_r$ iff $h(b) \in \Pi^i_r$ and $c
\in \Sigma^i_{l}$ iff $h(c) \in \Pi^i_{l}$
  \item $h$ preserves event clock partition $i.e.$ $\rho_1(a) = \rho_2(h(a))$ for all $a \in \Sigma$.
 \end{itemize}
\end{paragraph}
 
\end{document}